\documentclass[12pt]{article}

\usepackage{amsmath}
\usepackage{amsthm}
\usepackage[margin=1in]{geometry}
\usepackage{thumbpdf,lmodern}
\usepackage{makecell}
\usepackage{tabularx}
\usepackage{graphicx}
\usepackage{multirow}
\usepackage[round]{natbib}
\usepackage{amsfonts}
\usepackage{framed}
\usepackage{hyperref}
\usepackage{authblk}
\usepackage{tablefootnote}
\usepackage{float}

\providecommand{\keywords}[1]
{
  \small	
  \textbf{\textit{Keywords---}} #1
}

\begin{document}

\title{Optimal Priors for the Discounting Parameter of the Normalized Power Prior}

\author[1]{Yueqi Shen\thanks{ys137@live.unc.edu}}
\author[2]{Luiz M. Carvalho}
\author[3]{Matthew A. Psioda}
\author[1]{Joseph G. Ibrahim}
\affil[1]{Department of Biostatistics, University of North Carolina at Chapel Hill}
\affil[2]{School of Applied Mathematics, Getulio Vargas Foundation}
\affil[3]{GSK}

\maketitle

\begin{abstract}
The power prior is a popular class of informative priors for incorporating information from historical data. It involves raising the likelihood for the historical data to a power, which acts as discounting parameter. When the discounting parameter is modelled as random, the normalized power prior is recommended. In this work, we prove that the marginal posterior for the discounting parameter for generalized linear models converges to a point mass at zero if there is any discrepancy between the historical and current data, and that it does not converge to a point mass at one when they are fully compatible. In addition, we explore the construction of optimal priors for the discounting parameter in a normalized power prior. In particular, we are interested in achieving the dual objectives of encouraging borrowing when the historical and current data are compatible and limiting borrowing when they are in conflict. We propose intuitive procedures for eliciting the shape parameters of a beta prior for the discounting parameter based on two minimization criteria, the Kullback-Leibler divergence and the mean squared error. Based on the proposed criteria, the optimal priors derived are often quite different from commonly used priors such as the uniform prior.
\end{abstract}

\section{Introduction}

The power prior \citep{chen_2000} is a popular class of informative priors that allow the incorporation of historical data through a tempering of the likelihood.
It is constructed by raising the historical data likelihood to a power $a_0$, where $0 \le a_0 \le 1$.
The discounting parameter $a_0$ can be fixed or modelled as random.
When it is modelled as random and estimated jointly with other parameters of interest, the normalized power prior (NPP) \citep{duan_2006} is recommended as it appropriately accounts for the normalizing function necessary for forming the correct joint prior distribution \citep{neuens_2009}. Many extensions of the power prior and the normalized power prior have been developed.
\cite{Banbeta_2019} develop the dependent and robust dependent normalized power priors which allow dependent discounting parameters for multiple historical datasets.
When the historical data model contains only a subset of covariates currently of interest and the historical information may not be equally informative for all parameters in the current analysis, \cite{Boonstra_Barbaro_2020} propose an extension of the power prior that adaptively combines a prior based upon the historical information with a variance-reducing prior that shrinks parameter values toward zero.

The power prior and the normalized power prior have been shown to have several desirable properties.
\cite{ibrahim_2003} show that the power prior defines an optimal class of priors in the sense that it minimizes a convex combination of Kullback-Leibler (KL) divergences between a distribution based on no incorporation of historical data and a distribution based on completely pooling the historical and current data.
\cite{YE202229} prove that the normalized power prior minimizes the expected weighted KL divergence similar to the one in \cite{ibrahim_2003} with respect to the marginal distribution of the discounting parameter.
They also prove that if the prior on $a_0$ is non-decreasing and if the difference between the sufficient statistics of the historical and current data is negligible from a practical standpoint, the marginal posterior mode of $a_0$ is close to one.
\cite{Carvalho_Ibrahim_2021} show that the normalized power prior is always well-defined when the initial prior is proper, and that, viewed as a function of the discounting parameter, the normalizing function is a smooth and strictly convex function.
\cite{Neelon_OMalley_2010} show through simulations that for large datasets, the normalized power prior may result in more downweighting of the historical data than desired.
\cite{Han_Ye_Wang_2022} point out that the normalizing function might be infinite with improper initial priors on the parameters of interest for $a_0$ values close to zero, in which case the admissible set of the discounting parameter excludes values close to zero. \cite{Pawel_2023} derive the marginal posterior distribution of $a_0$ when a beta prior is used for $a_0$ for \emph{i.i.d.} normal and binomial models, and show that, under these model assumptions, the marginal posterior of $a_0$ does not converge to a point mass at one when the sample size becomes arbitrarily large. In this paper, we provide a formal proof of the limiting behavior of the marginal posterior of $a_0$ for generalized linear models (GLMs) when the historical and current datasets are fully compatible and when there is discrepancy.


Many empirical Bayes-type approaches have been developed to adaptively determine the discounting parameter.
For example, Gravestock and Held \citep{Gravestock_Held_2017,Gravestock_Held_2019} propose to set $a_0$ to the value that maximizes the marginal likelihood.
\cite{Liu_2018} proposes choosing $a_0$ based on the p-value for testing the compatibility of the current and historical data.
\cite{Bennett_2021} propose using an equivalence probability weight and a weight based on tail area probabilities to assess the degree of agreement between the historical and current control data for cases with binary outcomes.
\cite{Pan_Yuan_Xia_2017} propose the calibrated power prior, where $a_0$ is defined as a function of a congruence measure between the historical and current data.
The function which links $a_0$ and the congruence measure is prespecified and calibrated through simulation. While these empirical Bayes approaches shed light on the choice of $a_0$, there has not been any fully Bayesian approach based on an optimal prior on $a_0$.

In this work, we first explore the asymptotic properties of the normalized power prior when the historical and current data are fully compatible (i.e., the sufficient statistics of the two datasets are equal) or incompatible (i.e., the sufficient statistics of the two datasets have some non-zero difference).
We prove that for GLMs utilizing a normalized power prior, the marginal posterior distribution of $a_0$ converges to a point mass at zero if there is any discrepancy between the historical and current data.
When the historical and current data are fully compatible, the asymptotic distribution of the marginal posterior of $a_0$ is derived for GLMs; we note that it does not concentrate around one. However, we prove an interesting finding that, for an i.i.d. normal model with finite samples, the marginal posterior of $a_0$ always has more mass around one when the datasets are fully compatible, compared to the case where there is any discrepancy. 
Secondly, we propose a novel fully Bayesian approach to elicit the shape parameters of the beta prior on $a_0$ based on two optimality criteria, Kullback-Leibler (KL) divergence and mean squared error (MSE). For the first criterion, we propose as optimal the beta prior whose shape parameters result in a minimized weighted average of KL divergences between the marginal posterior for $a_0$ and user-specified target distributions based on hypothetical scenarios where there is no discrepancy and where there is a maximum tolerable discrepancy. This class of priors on $a_0$ based on the KL criterion is optimal in the sense that it is the best possible beta prior at balancing the dual objectives of encouraging borrowing when the historical and current data are compatible and limiting borrowing when they are in conflict. For the second criterion, we propose as optimal the beta prior whose shape parameters result in a minimized weighted average of the MSEs based on the posterior mean of the parameter of interest when its hypothetical true value is equal to its estimate using the historical data, or when it differs from its estimate by the maximum tolerable amount. We study the properties of the proposed approaches \textit{via} simulations for the \textit{i.i.d.} normal and Bernoulli cases as well as for the normal linear model.
Two real-world case studies of clinical trials with binary outcomes and covariates demonstrate the performance of the optimal priors compared to conventionally used priors on $a_0$, such as a uniform prior.  

\section{Asymptotic Properties of the Normalized Power Prior}\label{sec2}

Let $D$ denote the current data and $D_0$ denote the historical data.
Let $\theta$ denote the model parameters and $L(\theta|D)$ denote a general likelihood function.
The power prior \citep{chen_2000} is formulated as 
\begin{align*}
\pi(\theta|D_0, a_0) \propto L(\theta|D_0)^{a_0}\pi_0(\theta),
\end{align*}
where $0 \le a_0 \le 1$ is the discounting parameter which discounts the historical data likelihood, and $\pi_0(\theta)$ is the initial prior for $\theta$.
The discounting parameter $a_0$ can be fixed or modelled as random.
Modelling $a_0$ as random allows researchers to account for uncertainty when discounting historical data and to adaptively learn the appropriate level of borrowing.
\cite{duan_2006} propose the \emph{normalized power prior}, given by 
 \begin{align}\label{npp}
      \pi(\theta, a_0|D_0) = \pi(\theta|D_0, a_0)\pi(a_0) = \frac{L(\theta|D_0)^{a_0}\pi_0(\theta)}{c(a_0)}\pi(a_0),
 \end{align}
where $c(a_0)=\int L(\theta|D_0)^{a_0}\pi_0(\theta) d\theta$ is the normalizing function.
The normalized power prior is thus composed of a conditional prior for $\theta$ given $a_0$ and a marginal prior for $a_0$.

Ideally, the posterior distribution of $a_0$ with the normalized power prior would asymptotically concentrate around zero when the historical and current data are in conflict, and around one when they are compatible.
In this section, we study the asymptotic properties of the normalized power prior for the exponential family of distributions as well as GLMs.
Specifically, we are interested in exploring the asymptotic behaviour of the posterior distribution of $a_0$ when the historical and current data are incompatible and when they are compatible, respectively.

\subsection{Exponential Family}\label{sec2:expfam}

First, we study the asymptotic properties of the normalized power prior for the  exponential family of distributions.
The density of a random variable $Y$ in the one-parameter exponential family has the form 
\begin{align}\label{expfam}
    p(y|\theta) = q(y)\exp\left(y\theta-b(\theta)\right),
\end{align}
where $\theta$ is the canonical parameter and $q(\cdot)$ and $b(\cdot)$ are known functions.
Suppose $D=(y_1,\dots,y_n)$ is a sample of $n$ \textit{i.i.d.} observations from an exponential family distribution in the form of \eqref{expfam}.
The likelihood is then given by 
\begin{align*}
    L(\theta|D) = Q(D)\exp\left(\sum_{i=1}^ny_i\theta-nb(\theta)\right),
\end{align*}
where $Q(D) = \prod_{i=1}^{n}q(y_i)$.
Suppose $D_0=(y_{01},\dots,y_{0n_0})$ is a sample of $n_0$ \textit{i.i.d.} observations from the same exponential family.
The likelihood for the historical data raised to the power $a_0$ is 
\begin{align*}
  [L(\theta|D_0)]^{a_0} = Q(D_0)^{a_0}\exp\left(a_0\left[\sum_{i=1}^{n_0}y_{0i}\theta-n_0b(\theta)\right]\right),
\end{align*}
 where $Q(D_0) = \prod_{i=1}^{n_0}q(y_{0i})$. Using the normalized power prior defined in \eqref{npp}, the joint posterior of $\theta$ and $a_0$ is given by 
\begin{align*}
        \pi(\theta,a_0|D,D_0) &\propto L(\theta|D)\pi(\theta,a_0|D_0) = L(\theta|D)\frac{L(\theta|D_0)^{a_0}\pi_0(\theta)}{c(a_0)}\pi(a_0).
\end{align*}
The marginal posterior of $a_0$ is given by
\begin{align}
    \pi(a_0|D,D_0) = \int\pi(\theta,a_0|D,D_0)d\theta \propto \int L(\theta|D)\frac{L(\theta|D_0)^{a_0}\pi_0(\theta)}{c(a_0)}\pi(a_0)d\theta.\label{apost1}
\end{align}

With these calculations in place, the question now arises as to what prior should be given to $a_0$.
One commonly used class of priors on $a_0$ is the beta distribution \citep{chen_2000}.
Let $\alpha_0$ and $\beta_0$ denote the shape parameters of the beta distribution.
We first prove that the marginal posterior of $a_0$ \eqref{apost1} with $\pi(a_0)=\text{beta}(\alpha_0,\beta_0)$ converges to a point mass at zero for a fixed, non-zero discrepancy between $\bar{y}$ and $\bar{y}_0$. 

\newtheorem{thm}{Theorem}[section]
\newtheorem{lem}[thm]{Lemma}
\newtheorem{cor}{Corollary}[section]
\begin{thm}\label{nocov_th}
Suppose $y_1,\dots,y_n$ and $y_{01}, \dots, y_{n_0}$ are independent observations from the same exponential family distribution \eqref{expfam}. Let $\bar{y}=\frac{1}{n}\sum_{i=1}^n y_{i}$ and $\bar{y}_0 =\frac{1}{n_0}\sum_{i=1}^{n_0} y_{0i}$.
Suppose also that the difference in the estimates of the canonical parameter $\theta$ is fixed and equal to $\delta$, i.e., $|\dot{b}^{-1}(\bar{y})-\dot{b}^{-1}(\bar{y}_0)| = \delta$, and $\frac{n_0}{n}=r$, where $\delta > 0$ is finite and $r > 0$ is a constant, and $\dot{b}(\cdot)=\partial_{\theta}b(\cdot)$.
Then, the marginal posterior of $a_0$ using the normalized power prior \eqref{apost1} with a $\operatorname{beta}(\alpha_0, \beta_0)$ prior on $a_0$, where $\alpha > 0$ and $\beta > 0$, converges to a point mass at $0$.
That is, $\lim\limits_{n\rightarrow \infty} \frac{\int_0^{\epsilon}\pi(a_0|D, D_0,\alpha_0,\beta_0)da_0}{\int_0^1 \pi(a_0|D, D_0,\alpha_0,\beta_0)da_0} = 1$ for any $\epsilon > 0$.
\end{thm}
\begin{proof}
See Appendix \ref{nocov_proof}.
\end{proof}

In Theorem \ref{nocov_th}, we fix the ratio of $n_0$ to $n$, denoted as $r$, and let $n$ go to infinity. This is because if $n_0$ is fixed while $n$ is allowed to go to infinity, then the historical likelihood in the normalized power prior becomes irrelevant for inference on $\theta$, regardless of the prior on $a_0$. Therefore, $n_0$ must go to infinity as well. Theorem \ref{nocov_th} asserts that the normalized power prior is sensitive to any discrepancy between the sufficient statistics in large samples, as the mass of the marginal distribution of $a_0$ will concentrate near zero as the sample size increases for any fixed difference $\delta$. Figure \ref{app_fig_marg} in Appendix \ref{app_sec2} shows that the marginal posterior of $a_0$ converges to a point mass at zero rather quickly as the sample size grows when there is discrepancy. 

The natural question to then ask is whether Theorem \ref{nocov_th} has a sort of converse in that the posterior should concentrate around one under compatibility.  
We derive the asymptotic marginal posterior distribution of $a_0$ when $\bar{y}=\bar{y}_0$ and show that it does not converge to a point mass at one. 

\begin{cor}\label{nocov_coro}
Suppose $y_1,\dots,y_n$ and $y_{01}, \dots, y_{n_0}$ are independent observations from the same exponential family distribution \eqref{expfam}. Let $\bar{y}=\frac{1}{n}\sum_{i=1}^n y_{i}$ and $\bar{y}_0 =\frac{1}{n_0}\sum_{i=1}^{n_0} y_{0i}$.
Suppose $\bar{y}=\bar{y}_0$ and $\frac{n_0}{n}=r$ where $r > 0$ is a constant.
The marginal posterior of $a_0$ using the normalized power prior, as specified in \eqref{apost1}, converges to
\begin{equation*}
    \tilde{\pi}(a_0 | D, D_0) = \frac{\sqrt{\frac{ra_0}{ra_0+1}}\pi(a_0)}{\int_0^1 \sqrt{\frac{ra_0}{ra_0+1}}\pi(a_0) da_0}
\end{equation*}
as $n \rightarrow \infty$.
\end{cor}
\begin{proof}
See Appendix \ref{nocov_coro_proof}.
\end{proof}

Corollary \ref{nocov_coro} shows that the normalized power prior fails to fully utilize the historical data when the means of the historical data and the current data are equal for a generic, non-degenerate prior on $a_0$. It is worth noting that the adjustment to the prior $\sqrt{\frac{ra_0}{ra_0+1}}$ is maximized at $a_0 = 1$.
If $\pi(a_0)$ is chosen to be concentrated near one, then the marginal posterior of $a_0$ may be concentrated near one. 

We recognize that, in reality, $\bar{y}=\bar{y}_0$ will occur with probability zero. However, the point of this theorem is to show that even in the most extreme case where we have identical data, the marginal posterior of $a_0$ with the normalized power prior does not concentrate around one. Relaxing the equality will lead to the same conclusion. In fact, Theorem \ref{nocov_th} indicates that when the difference between $\bar{y}$ and $\bar{y}_0$ is fixed at some $\delta \neq 0$, the posterior of $a_0$ asymptotically converges to a point mass at zero.

Even though the marginal posterior of $a_0$ does not concentrate under one when $\bar{y}=\bar{y}_0$, we find that for the i.i.d. normal model, the cumulative density function when $\bar{y}=\bar{y}_0$ is dominated by every other cumulative density function, as demonstrated by the following theorem. This is a novel result which establishes a desirable borrowing property of the normalized power prior for i.i.d. normal data for finite samples. 
\begin{thm}\label{cdf}
Suppose $y_1,\dots,y_n \sim N(\theta, \sigma^2)$ and $y_{01}, \dots, y_{n_0} \sim N(\theta, \sigma_0^2)$, where $\sigma^2$, $\sigma_0^2$, $n$ and $n_0$ are fixed. Let $\bar{y}=\frac{1}{n}\sum_{i=1}^n y_{i}$ and $\bar{y}_0 =\frac{1}{n_0}\sum_{i=1}^{n_0} y_{0i}$. Let $|\bar{y}-\bar{y}_0|=d$. Let $F_{d}(a_0)$ denote cumulative density function of the marginal posterior of $a_0$ using the normalized power prior in \eqref{apost1}, and let $w_{d}(a_0)=F_{d}(a_0)-F_0(a_0)$. The function $w_{d}(a_0) > 0$ for all $d > 0$ and $0 < a_0 < 1$.
\end{thm}
\begin{proof}
See Appendix \ref{cdf_proof}.
\end{proof}
Theorem \ref{cdf} shows that for i.i.d. normal data, for any fixed sample sizes, $1-F_0(a_0)$ will always be greater than $1-F_d(a_0)$ for any $d > 0$. This means the posterior of $a_0$ always has more mass around one when the datasets are fully compatible, compared to the case where there is any discrepancy, which is a desirable property. 

\subsection{Generalized Linear Models}

The ability to deal with non \textit{i.i.d.} data and incorporate covariates is crucial to the applicability of the normalized power prior; we thus now extend these results to generalized linear models (GLMs).
We first define the GLM with a canonical link and fixed dispersion parameter.
Let $y_i$ denote the response variable and $x_i$ denote a $p$-dimensional vector of covariates for subject $i=1, \dots, n$.
Let $\beta = (\beta_1, \dots, \beta_p)'$ be a $p$-dimensional vector of regression coefficients.
The GLM with a canonical link is given by 
\begin{align}\label{glm}
    p(y_i|x_i, \beta, \phi) = q(y_i, \phi)\exp\{\phi^{-1}[y_ix_i'\beta-b(x_i'\beta)]\}.
\end{align}
Without loss of generality, we assume $\phi=1$.
Let $D=\{(y_i,x_i), i=1,\dots,n\}\equiv(n,Y_{n\times 1},X_{n\times p})$ where $Y=(y_1,\dots,y_n)'$ and $X=(x_1,\dots,x_n)'$.
Assuming the $y_i$'s are (conditionally) independent, the likelihood is given by 
\begin{align*}
    L(\beta|D) = Q(Y)\exp\left(\sum_{i=1}^n y_i x_i'\beta-\sum_{i=1}^nb(x_i'\beta)\right),
\end{align*}
where $Q(Y) = \prod_{i=1}^{n}q(y_i, 1)$.
Let $\hat{\beta}$ denote the posterior mode of $\beta$ obtained by solving $\partial_{\beta}\log L(\beta|D)=0$.
Let $D_0=\{(y_{0i},x_{0i}), i=1,\dots,n_0\}\equiv(n_0,Y_{0n_0\times 1},X_{0n_0\times p})$ where $Y_0=(y_{01},\dots,y_{0n_0})'$ and $X_0=(x_{01},\dots,x_{0n_0})'$.
Assuming the $y_{0i}$'s are (conditionally) independent, the historical data likelihood raised to the power $a_0$ is given by
\begin{align*}
    [L(\beta|D_0)]^{a_0} = Q(Y_0)^{a_0}\exp\left(a_0\left[\sum_{i=1}^{n_0}y_{0i}x_{0i}'\beta-\sum_{i=1}^nb(x_{0i}'\beta)\right]\right),
\end{align*}
 where $Q(Y_0) = \prod_{i=1}^{n_0}q(y_{0i}, 1)$.
 Let $c^*(a_0)=\int L(\beta|y_0)^{a_0}\pi_0(\beta) d\beta$.
 Using the normalized power prior defined in \eqref{npp}, the joint posterior of $\beta$ and $a_0$ is given by 
\begin{align*}
        \pi(\beta,a_0|D,D_0) &\propto L(\beta|D)\pi(\beta,a_0|D_0) = L(\beta|D)\frac{L(\beta|D_0)^{a_0}\pi_0(\beta)}{c^*(a_0)}\pi(a_0).
\end{align*}
Let $\hat{\beta}_0$ denote the posterior mode of $\beta$ obtained by solving $\partial_{\beta}\log \left[\frac{L(\beta|D_0)^{a_0}\pi_0(\beta)}{c^*(a_0)}\right]=0.$ Note that $\hat{\beta}_0$ is independent of $a_0$ after we take the limit. The marginal posterior of $a_0$ is given by
\begin{align}
    \pi(a_0|D,D_0) &= \int\pi(\beta,a_0|D,D_0)d\beta \propto \int L(\beta|D)\frac{L(\beta|D_0)^{a_0}\pi_0(\beta)}{c^*(a_0)}\pi(a_0)d\beta.\label{glmnpp}
\end{align}
Now we extend Theorem \ref{nocov_th} to GLMs. 
\begin{thm}\label{glm_th}
Suppose $X$ is $n\times p$ of rank $p$ and $X_0$ is $n_0\times p$ of rank $p$. Suppose $\hat{\beta}-\hat{\beta}_0 = \delta$ where $\delta \neq 0$ is a finite vector, and $\frac{n_0}{n}=r$ where $r > 0$ is a constant scalar.
Assume $n\left[\frac{\partial^2\log [L(\beta|D)]}{\partial\beta_i\partial\beta_j}\right]^{-1}$ and $n_0a_0\left[\frac{\partial^2\log [L(\beta|D_0)^{a_0}\pi_0(\beta)]}{\partial\beta_i\partial\beta_j}\right]^{-1}$ do not depend on $n$ and $a_0$.
Then, the marginal posterior of $a_0$ using the normalized power prior \eqref{glmnpp} with a $\operatorname{beta}(\alpha_0, \beta_0)$ prior on $a_0$, where $\alpha > 0$ and $\beta > 0$, converges to a point mass at zero.
That is, $\lim\limits_{n\rightarrow \infty} \frac{\int_0^{\epsilon}\pi(a_0|D, D_0, \alpha_0, \beta_0)da_0}{\int_0^1 \pi(a_0|D, D_0, \alpha_0, \beta_0)da_0} = 1$ for any $\epsilon > 0$.
\end{thm}

\begin{proof}
See Appendix \ref{glm_proof}. 
\end{proof}

Theorem \ref{glm_th} asserts that the normalized power prior is sensitive to discrepancies in the historical and current data in the presence of covariates.
The mass of the marginal distribution of $a_0$ will concentrate near zero as the sample size increases for any fixed discrepancy between the historical and current data, assuming $\frac{1}{n}X'X$ and $\frac{1}{n_0}X_0'X_0$ are fixed, i.e., $n\left[\frac{\partial^2\log [L(\beta|D)]}{\partial\beta_i\partial\beta_j}\right]^{-1}$ and $n_0a_0\left[\frac{\partial^2\log [L(\beta|D_0)^{a_0}\pi_0(\beta)]}{\partial\beta_i\partial\beta_j}\right]^{-1}$ do not depend on $n$ and $a_0$. 

Next, we derive the asymptotic marginal posterior distribution of $a_0$ when the sufficient statistics and covariate (design) matrices of the historical and current data equal. 


\begin{cor}\label{glm_coro}
Suppose $X$ is $n\times p$ of rank $p$ and $X_0$ is $n_0\times p$ of rank $p$.
Let $Y=(y_1,\dots,y_n)'$ and $Y_0=(y_{01},\dots,y_{0n_0})'$. Consider the GLM in \eqref{glm}.
If $n=n_0$, $X=X_0$, and $X'Y = X_0'Y_0$, then the marginal posterior of $a_0$  using the normalized power prior, as specified in \eqref{glmnpp}, converges to
\begin{equation*}
    \tilde{\pi}(a_0 | X, Y, X_0, Y_0) = \frac{\left(\frac{a_0}{a_0+1}\right)^{\frac{p}{2}}\pi(a_0)}{\int_0^1 \left(\frac{a_0}{a_0+1}\right)^{\frac{p}{2}}\pi(a_0) da_0},
\end{equation*}
as $n \rightarrow \infty$.
\end{cor}
\begin{proof}
See Appendix \ref{glm_coro_proof}.
\end{proof}

Corollary \ref{glm_coro} states that the marginal posterior of $a_0$ using the normalized power prior does not converge to a point mass at one when the sufficient statistics and the covariates of the historical and current data are equal. We also observe that as $p$ approaches infinity, the marginal posterior of $a_0$ specified above converges to a point mass at one. The form of the asymptotic marginal posterior of $a_0$ suggests that the normalized power prior may be sensitive to overfitting when the historical and current datasets are compatible. Figure \ref{app_fig_coro} in Appendix \ref{app_sec2} presents numerical results corroborating Corollaries \ref{nocov_coro} and \ref{glm_coro}. Therein we present histograms of posterior samples of $a_0$ that faithfully recapitulate the theoretical density functions. 

In Theorem~\ref{glm_comp} we also relax the previous result by deriving the asymptotic marginal posterior distribution of $a_0$ assuming only that the sufficient statistics of the historical and current data are equal.
This means that the covariate matrices need not be equal so long as the sufficient statistics $X'Y$ and $X_0'Y_0$ are, increasing the applicability of the result.


\begin{thm}\label{glm_comp}
Suppose $X$ is $n\times p$ of rank $p$ and $X_0$ is $n_0\times p$ of rank $p$. Let $Y=(y_1,\dots,y_n)'$ and $Y_0=(y_{01},\dots,y_{0n_0})'$.
Consider the GLM in \eqref{glm}, where $\frac{n_0}{n}=r$ and $r > 0$ is a constant.
If $X'Y = X_0'Y_0$ and $X\neq X_0$, then the marginal posterior of $a_0$ using the normalized power prior, as specified in \eqref{glmnpp}, is asymptotically proportional to $$\pi(a_0)\cdot\frac{|\hat{\Sigma}_g|^{1/2}}{|\tilde{\Sigma}_k|^{1/2}}\exp\left\{-n[g_n(\hat{\beta})-k_n(\tilde{\beta})]\right\},$$
where the definitions of $g_n(\beta)$, $k_n(\beta)$ and $\frac{|\hat{\Sigma}_g|^{1/2}}{|\tilde{\Sigma}_k|^{1/2}}$ can be found in Appendix \ref{glm_comp_proof}. 
\end{thm}
\begin{proof}
See Appendix \ref{glm_comp_proof}.
\end{proof}

Corollary \ref{glm_coro} and Theorem \ref{glm_comp} show that, for GLMs, the marginal posterior of $a_0$ using the normalized power prior does not converge to a point mass at one when the sufficient statistics of the historical and current data are equal.
From Theorems \ref{nocov_th}-\ref{glm_comp}, we conclude that, asymptotically, the normalized power prior is sensitive to discrepancies between the historical and current data, but cannot fully utilize the historical information when there are no discrepancies. However, we show that the posterior of $a_0$ always has the most mass around one when the datasets are fully compatible for finite i.i.d. normal observations. 

We highlight the differences between the theorems above and the results presented in \cite{Pawel_2023}. \cite{Pawel_2023} derive the marginal posterior distribution of $a_0$ when a beta prior is used for $a_0$ for \emph{i.i.d.} normal and binomial models and show through graphical approaches that the distribution shifts toward zero as the standard error of the current data converges to zero, when $\hat{\theta} \neq \hat{\theta}_0$, and the standard error of the historical data is fixed. We prove this phenomenon much more generally and analytically for the exponential family of distributions as well as GLMs.

\section{Optimal Beta Priors for $a_0$}
\subsection{Kullback-Leibler Divergence Criterion}\label{sec3}

In this section, we propose a prior based on minimizing the KL divergence of the marginal posterior of $a_0$ to two reference distributions. This resulting prior is optimal in the sense that it is the best possible beta prior at balancing the dual objectives of encouraging borrowing when the historical and current data are compatible and limiting borrowing when they are in conflict.

Let $\bar{y}_0$ denote the mean of the historical data and $\bar{y}$ denote the mean of the hypothetical current data.
Let $\pi_1(a_0) \equiv \operatorname{beta}(c, 1)$ ($c \gg 1$ is fixed) and $\pi_2(a_0) \equiv \operatorname{beta}(1, c)$.
The distributions $\pi_1(a_0)$ and $\pi_2(a_0)$ represent two ideal scenarios, where $\pi_1(a_0)$ is concentrated near one and $\pi_2(a_0)$ is concentrated near zero.
The KL-based approach computes the hyperparameters ($\alpha_0$ and $\beta_0$) for the beta prior on $a_0$ that will minimize a convex combination of two KL divergences; one is the KL divergence between $\pi_1(a_0)$ and the marginal posterior of $a_0$ when $\bar{y}=\bar{y}_0$, while the other is the KL divergence between $\pi_2(a_0)$ and the marginal posterior of $a_0$ when there is a user-specified difference between $\bar{y}$ and $\bar{y}_0$.

Let $d=\bar{y} - \bar{y}_0$, representing the difference between the means of the hypothetical current data and the historical data.
Our approach is centered on a user-specified \textbf{maximum tolerable difference} (MTD), $d_{\textrm{MTD}}$. Let $\pi^*(a_0)$ denote the marginal posterior of $a_0$ when $d=0$. Let $\pi_{\textrm{MTD}}(a_0)$ denote the marginal posterior of $a_0$ when $d=d_{\textrm{MTD}}$. For $d=0$, we want $\pi^*(a_0)$ to resemble $\pi_1(a_0)$ and for $d = d_{\textrm{MTD}}$, we want $\pi_{\textrm{MTD}}(a_0)$ to resemble $\pi_2(a_0)$.
The distributions $\pi_1(a_0)$ and $\pi_2(a_0)$ have been chosen to correspond to cases with substantial and little borrowing, respectively.
Therefore, our objective is to solve for $\alpha_0>0$ and $\beta_0>0$ to minimize
\begin{equation*}
    K(\alpha_0, \beta_0) = w KL(\pi^*(a_0), \pi_1(a_0)) + (1-w)KL(\pi_{\textrm{MTD}}(a_0), \pi_2(a_0)).
\end{equation*}
Here $0 < w < 1$ is a scalar and $KL(p,q)$ for distributions $P$ and $Q$ with P as reference is defined as
\begin{align*}
KL(p,q) &= \int \log\left(\frac{p(x)}{q(x)}\right)dP(x)= E_p[\log(p)] - E_p[\log(q)].
\end{align*}
The scalar $w$ weights the two competing objectives. For $w > 0.5$, the objective to encourage borrowing is given more weight, and for $w < 0.5$, the objective to limit borrowing is given more weight. Even though this approach requires specifying $w$, $c$ and $d_{\textrm{MTD}}$, these parameters are much more intuitive to choose than the hyperparameters of the prior for $a_0$.

Below we demonstrate the simulation results using this method for the \textit{i.i.d.} normal case, the \textit{i.i.d.} Bernoulli case and the normal linear model.
We compare the marginal posterior of $a_0$ using the KL-based optimal prior with that using the uniform prior.
For all simulations in this section, we choose $w=0.5$ so that the two competing objectives are given equal weight.
We choose $c=10$ so that $\pi_1(a_0)$ and $\pi_2(a_0)$ represent cases with substantial and little borrowing, respectively. We vary the choice of $d_{\textrm{MTD}}$ and examine each resulting optimal prior. In practice, we recommend that $d_{\textrm{MTD}}$ be no larger than the treatment effect of the historical study.

\subsubsection{Normal \textit{i.i.d.} Case}

We assume $y_1,\dots,y_n$ and $y_{01},\dots,y_{0{n_0}}$ are \textit{i.i.d.} observations from N($\mu$, $\sigma^2$) where $\sigma^2=1$.
We choose $\bar{y}_0=1.5$ and $n=n_0=30$.
The objective function $K(\cdot, \cdot)$ is computed using numerical integration and optimization is performed using the \verb|optim()| function in (base) R \citep{r}.

\begin{figure}[h]
\begin{center}
\includegraphics[width=13cm]{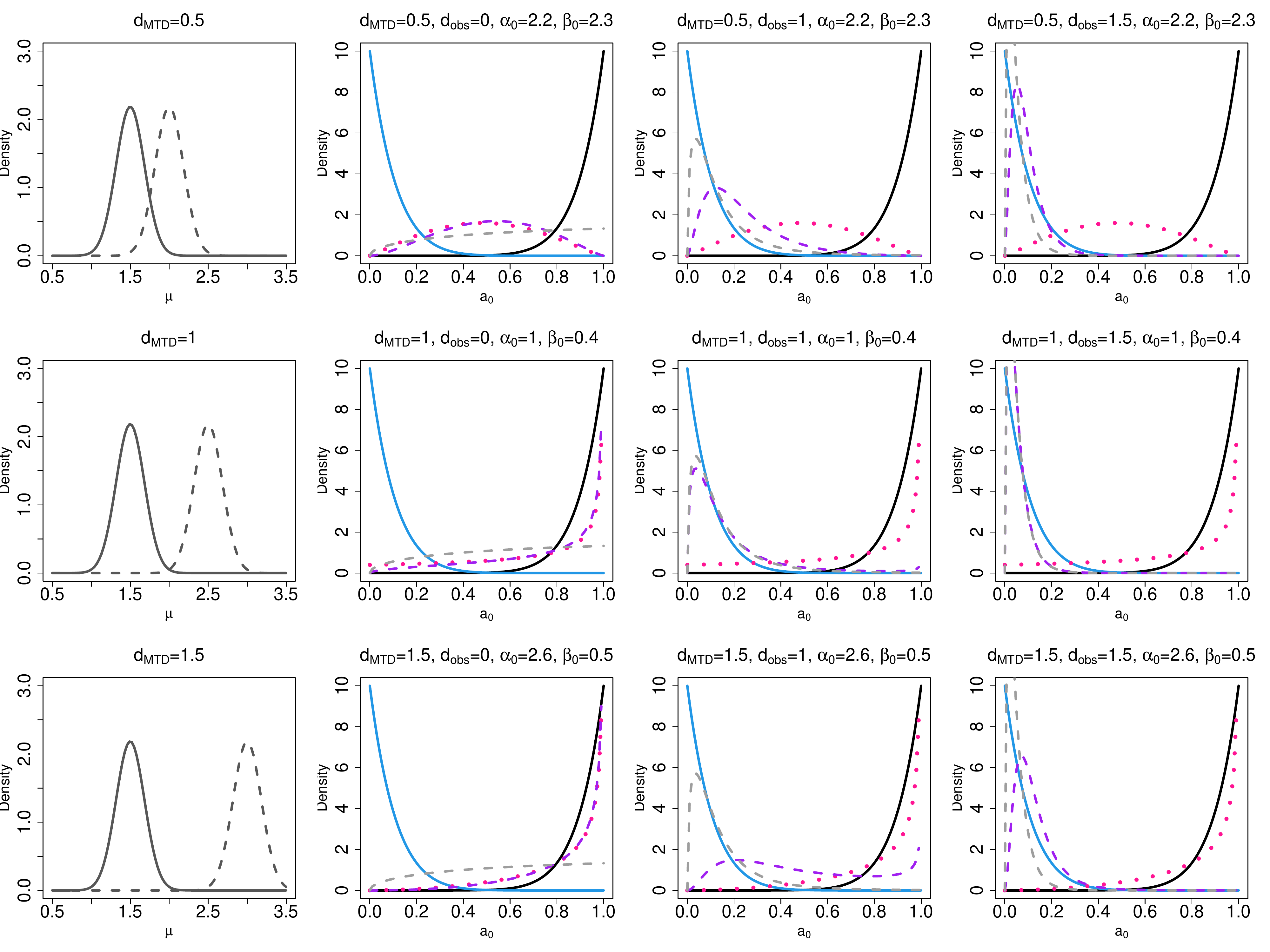}
\end{center}
\caption{Simulation results for the normal \textit{i.i.d.} case, where $\sigma^2=1$, $\bar{y}_0=1.5$ and $n=n_0=30$.
The first figure of each row plots the historical (black solid line) and current (black dashed line) data likelihoods if the hypothetical degree of conflict is equal to $d_{\textrm{MTD}}$.
For each row of the figure, the maximum tolerable difference $d_{\textrm{MTD}}$ is chosen to be $0.5$, $1$ and $1.5$, and the corresponding optimal prior (pink dotted line) is derived for each value of $d_{\textrm{MTD}}$.
For each optimal prior, we vary $d_{\textrm{obs}}=\bar{y}_{\textrm{obs}}-\bar{y}_0$ to evaluate the performance of the optimal prior for different observed data.
For columns 2-4, $d_{\textrm{obs}}$ is chosen to be $0$, $1$ and $1.5$, respectively.
The black and blue curves correspond to $\pi_1(a_0)\equiv\operatorname{beta}(10, 1)$ and $\pi_2(a_0)\equiv\operatorname{beta}(1, 10)$, respectively.
The purple dashed line represents the marginal posterior of $a_0$ with the optimal prior for a given $d_{\textrm{obs}}$.
The grey dashed line plots the marginal posterior of $a_0$ with the uniform prior.}
\label{sim_normal}
\end{figure}

In Figure~\ref{sim_normal}, the first figure of each row plots the historical and current data likelihoods if the hypothetical degree of conflict is equal to $d_{\textrm{MTD}}$.
For each row of the figure below, the maximum tolerable difference $d_{\textrm{MTD}}$ is chosen to be $0.5$, $1$ and $1.5$, and the corresponding optimal prior is derived for each value of $d_{\textrm{MTD}}$.
For each optimal prior, we vary the observed sample mean, denoted by $\bar{y}_{\textrm{obs}}$, to evaluate the posterior based on the optimal prior for different observed current data.
We use $d_{\textrm{obs}}=\bar{y}_{\textrm{obs}}-\bar{y}_0$ to represent the difference between the means of the observed current data and the historical data.
For columns 2-4,  $d_{\textrm{obs}}$ is chosen to be $0$, $1$ and $1.5$, respectively. Note that the values of $d_{\textrm{MTD}}$ and $d_{\textrm{obs}}$ are relative to the choices of $\sigma^2$, $n$ and $n_0$. For example, for larger $n$, $d_{\textrm{MTD}}$ would need to be decreased to produce a similar plot to Figure~\ref{sim_normal}.


From columns 2-4, we observe that when $d_{\textrm{MTD}}=0.5$, very little conflict is tolerated, and the resulting optimal prior does not strongly encourage either borrowing substantially or borrowing little.
As $d_{\textrm{MTD}}$ becomes larger, larger conflict is allowed and the optimal prior shifts more towards $\pi_1(a_0)$.
We also observe that when $d_{\textrm{MTD}}=1$ (the optimal hyperparameters are $\alpha_0=1$ and $\beta_0=0.4$) and $d_{\textrm{MTD}}=1.5$ (the optimal hyperparameters are $\alpha_0=2.6$ and $\beta_0=0.5$), the marginal posterior of $a_0$ with the optimal prior more closely mimics the target distribution when $d_{\textrm{obs}}=0$, i.e., the observed current and historical data are fully compatible.
As $d_{\textrm{obs}}$ increases, the marginal posterior shifts toward zero.
This behaviour is highly desirable as it achieves both goals of encouraging borrowing when the datasets are compatible and limiting borrowing when they are incompatible.

We can compare the marginal posterior of $a_0$ using the optimal prior with that using a uniform prior in Figure \ref{sim_normal}.
We observe that while the marginal posterior on $a_0$ with the uniform prior is very responsive to conflict, it does not concentrate around one even when the datasets are fully compatible.
We conclude that when $d_{\textrm{MTD}}$ is chosen to be reasonably large, the optimal prior on $a_0$ achieves a marginal posterior that is close to the target distribution when the datasets are fully compatible, while remaining responsive to conflict in the data.

\begin{table}
\begin{center}
\caption{Posterior mean (variance) of $\mu$ for the normal \textit{i.i.d.} case}\label{tab_normal}
\begin{tabular}{lcccc}
 \\
& $d_{\textrm{obs}}=0$ & $d_{\textrm{obs}}$=0.5 & $d_{\textrm{obs}}$=1 & $d_{\textrm{obs}}$=1.5\\[5pt]
\hline
$d_{\textrm{MTD}}$=0.5 & 1.5 (0.022) & 1.85 (0.026) & 2.32 (0.036) & 2.87 (0.036) \\
$d_{\textrm{MTD}}$=1 & 1.5 (0.019) & 1.82 (0.026) & 2.36 (0.042)& 2.93 (0.035)\\
$d_{\textrm{MTD}}$=1.5 & 1.5 (0.018) & 1.78 (0.020) & 2.19 (0.040) & 2.84 (0.039)\\
\end{tabular}
\end{center}
\end{table}

Table \ref{tab_normal} shows the posterior mean and variance of the mean parameter $\mu$ for various combinations of $d_{\textrm{MTD}}$ and $d_{\textrm{obs}}$ values corresponding to the  scenarios in Figure \ref{sim_normal}. The posterior mean and variance of $\mu$ with a normalized power prior are computed using the R package \textit{BayesPPD} \citep{bayesppd}. 
Again, $\bar{y}_0$ is fixed at $1.5$. Since $\bar{y}_{\textrm{obs}} \ge \bar{y}_0$, within each row, the posterior mean of $\mu$ is always smaller than $\bar{y}_{\textrm{obs}}$ due to the incorporation of $\bar{y}_0$.
We can also compare the results by column.
For fixed $d_{obs}$ (or equivalently $\bar{y}_{\textrm{obs}}$), if more historical information is borrowed, we expect the posterior mean of $\mu$ to be smaller.
When $d_{\textrm{obs}}=0$, the posterior mean stays constant while the variance decreases as $d_{\textrm{MTD}}$ increases.
If the maximum tolerable difference is large, more historical information is borrowed, leading to reduced variance. When $d_{\textrm{obs}}=0.5$, the posterior of $\mu$ decreases as more borrowing occurs when $d_{\textrm{MTD}}$ increases.
When $d_{\textrm{obs}}=1$ or $1.5$, the posterior of $\mu$ first increases and then decreases, as $d_{\textrm{MTD}}$ increases.
This is a result of two competing phenomena interacting; as $d_{\textrm{MTD}}$ increases, the optimal prior gravitates towards encouraging borrowing; however, since $d_{\textrm{obs}}$ is very large, the marginal posterior of $a_0$ moves toward zero even though the prior moves toward one.
In conclusion, we argue that the posterior estimates of $\mu$ with the optimal prior respond in a desirable fashion to changes in the data. In Appendix \ref{app_num}, we have included numerical experiments that demonstrate the reliability of the optimization scheme.



\subsubsection{Bernoulli Model}

For the Bernoulli model, we assume $y_1,\dots,y_n$ and $y_{01},\dots,y_{0{n_0}}$ are \textit{i.i.d.} observations from a Bernoulli distribution with mean $\mu$.
Again, we choose $n=n_0=30$ and optimization is performed analogously to the normal case. 

The resulting optimal priors and posteriors are shown in Figure \ref{sim_bern}. For each row of Figure \ref{sim_bern} below, the maximum tolerable difference $d_{\textrm{MTD}}$ is chosen to be $0.2$, $0.4$ and $0.6$, and the corresponding optimal prior is derived for each value of $d_{\textrm{MTD}}$.
For each optimal prior, we vary the observed $\bar{y}_{\textrm{obs}}$ to evaluate the performance of the optimal prior for different observed data.
For columns 2-4, $d_{\textrm{obs}}=\bar{y}_{\textrm{obs}}-\bar{y}_0$ is chosen to be $0$, $0.4$ and $0.6$, respectively.
Values of $\bar{y}_0$ and $\bar{y}_{\textrm{obs}}$ are chosen so that the variance stays constant for different values of $d_{\textrm{MTD}}$ or $d_{\textrm{obs}}$.

\begin{figure}
\begin{center}
\includegraphics[width=13cm]{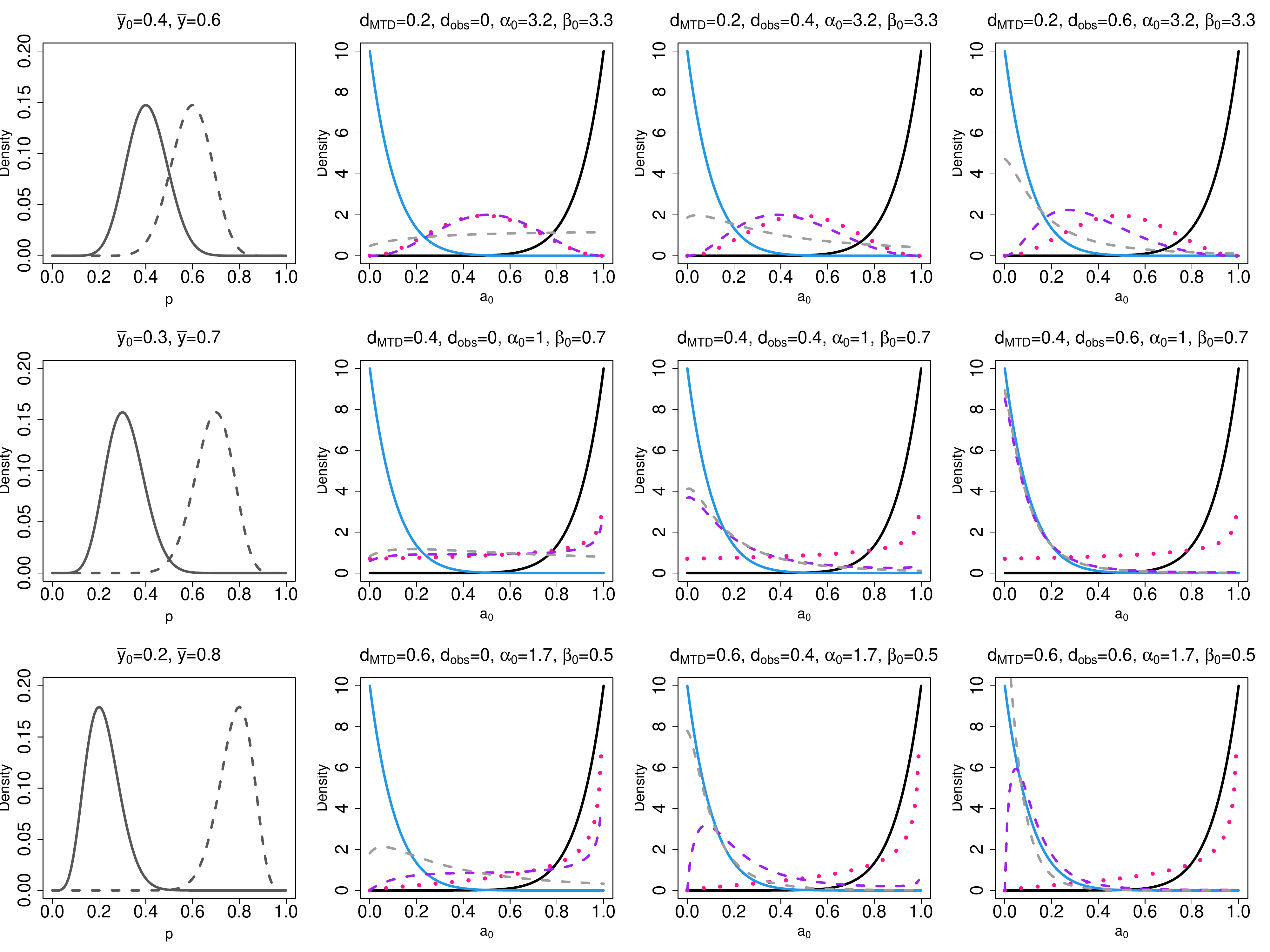}
\end{center}
\caption{Simulation results for the Bernoulli \textit{i.i.d.} case, where $\sigma^2=1$ and $n=n_0=30$. The first figure of each row plots the historical (black solid line) and current (black dashed line) data likelihoods if the hypothetical degree of conflict is equal to $d_{\textrm{MTD}}$.
For each row of the figure, the maximum tolerable difference $d_{\textrm{MTD}}$ is chosen to be $0.5$, $1$ and $1.5$, and the corresponding optimal prior (pink dotted line) is derived for each value of $d_{\textrm{MTD}}$.
For each optimal prior, we vary $d_{\textrm{obs}}=\bar{y}_{\textrm{obs}}-\bar{y}_0$ to evaluate the performance of the optimal prior for different observed data.
For columns 2-4, $d_{\textrm{obs}}$ is chosen to be $0$, $1$ and $1.5$, respectively.
The black and blue curves correspond to $\pi_1(a_0)\equiv\operatorname{beta}(10, 1)$ and $\pi_2(a_0)\equiv\operatorname{beta}(1, 10)$, respectively.
The purple dashed line represents the marginal posterior of $a_0$ with the optimal prior for a given $d_{\textrm{obs}}$.
The grey dashed line plots the marginal posterior of $a_0$ with the uniform prior.}
\label{sim_bern}
\end{figure}

The optimal marginal prior and posterior of $a_0$ for Bernoulli data are similar to those of the normal model. We observe that when the datasets are perfectly compatible, i.e., $d_{\textrm{obs}}=0$, the marginal posterior of $a_0$ with the optimal prior concentrates around one when $d_{\textrm{MTD}}$ is relatively large.
When $d_{\textrm{obs}}$ increases to $0.4$ or $0.6$, the marginal posterior of $a_0$ concentrates around zero when $d_{\textrm{MTD}}$ is relatively large.
The optimal prior becomes increasingly concentrated near one as $d_{\textrm{MTD}}$ increases.
Compared to the marginal posterior with the uniform prior, the optimal prior on $a_0$ achieves a marginal posterior that closely mimics the target distribution when the datasets are fully compatible, while remaining responsive to conflict in the data.

\subsubsection{Normal Linear Model}
Suppose $y_{01},\dots,y_{0{n_0}}$ are independent observations from the historical data where $y_{0i} \sim N(\beta_0 + \beta_1x_{0i}, \sigma^2)$ for the $i$-th observation and $x_{0i}$ is a single covariate.
Also suppose $y_1,\dots,y_n$ are independent observations from the current data where $y_j \sim N(\beta_0 + \beta_1x_j + d_{\textrm{MTD}}, \sigma^2)$ for the $j$-th observation and $x_j$ is a single covariate.
We vary $d_{\textrm{MTD}}$ to represent different degrees of departure of the intercept of the simulated current data to the intercept of the historical data.
We choose $\beta_0=1.5$, $\beta_1=-1$, $\sigma^2=1$ and $n=n_0=30$. We choose $d_{\textrm{MTD}}=0.1, 0.5$, and $1$ and $d_{\textrm{obs}}=0, 0.5$, and $1$.
The objective function $K$ is computed using Monte Carlo integration and optimization is performed using the \verb|optim()| function in R.

\begin{figure}
\begin{center}
\includegraphics[width=13cm]{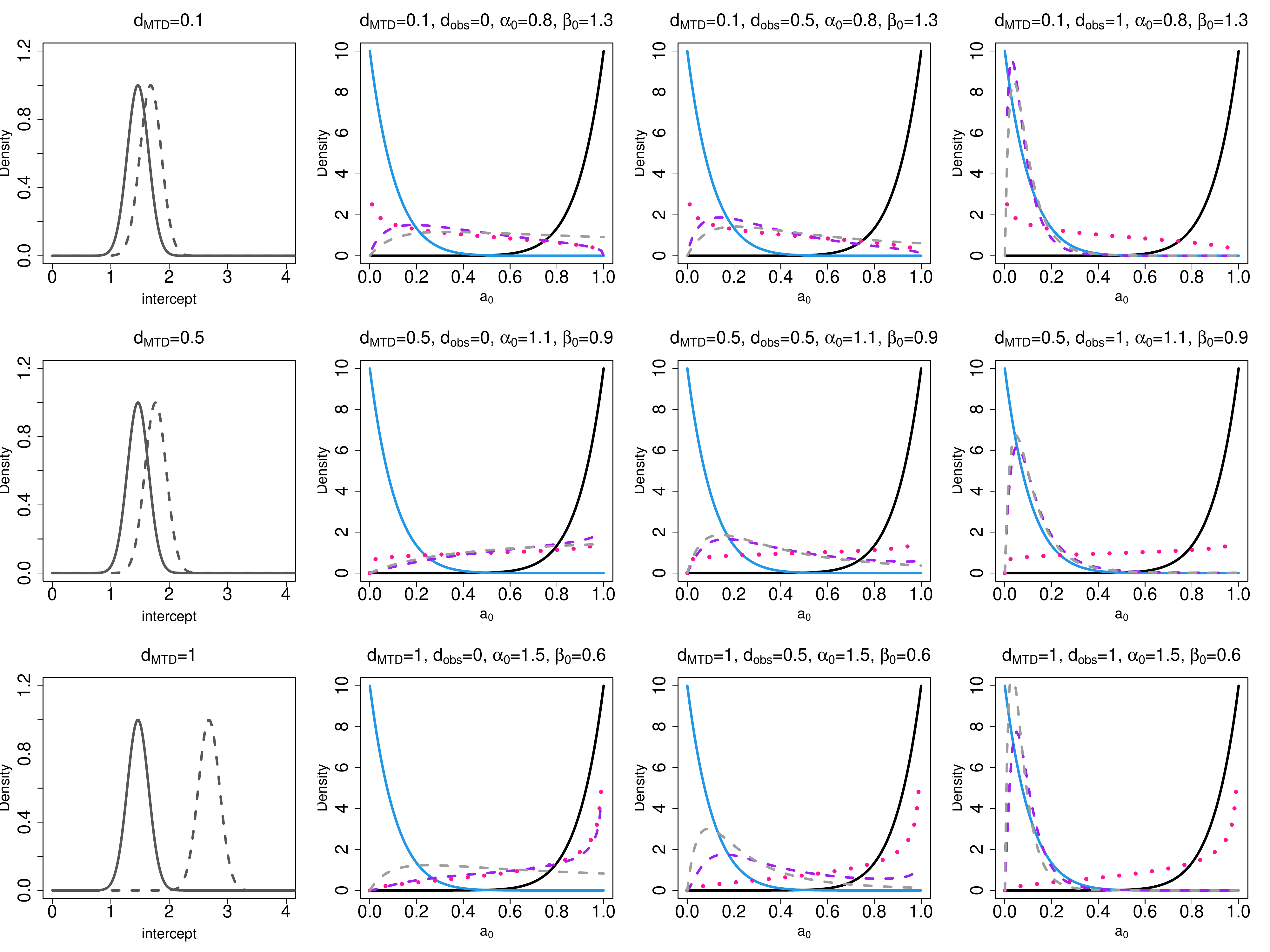}
\end{center}
\caption{Simulation results for the normal linear model with one covariate where $\beta_0=1.5$, $\beta_1=-1$, $\sigma^2=1$ and $n=n_0=30$.
The first figure of each row shows the historical (black solid line) and current (black dashed line) data likelihoods as a function of the intercept if the hypothetical degree of conflict is equal to $d_{\textrm{MTD}}$.
For each row, $d_{\textrm{MTD}}$ is chosen to be $0.1$, $0.5$ and $1$, and the corresponding optimal prior (pink dotted line) is derived for each value of $d_{\textrm{MTD}}$.
For each optimal prior, we vary $d_{\textrm{obs}}$ to represent different degrees of departure of the intercept of current data to that of historical data.
For columns 2-4, $d_{\textrm{obs}}$ is chosen to be $0$, $0.5$ and $1$, respectively.
The black and blue curves correspond to $\pi_1(a_0)\equiv\operatorname{beta}(10, 1)$ and $\pi_2(a_0)\equiv\operatorname{beta}(1, 10)$, respectively. 
The purple dashed line represents the marginal posterior of $a_0$ with the optimal prior for a given $d_{\textrm{obs}}$.
The grey dashed line plots the marginal posterior of $a_0$ with the uniform prior.}
\label{sim_nlm}
\end{figure}

Figure \ref{sim_nlm} shows the optimal prior and optimal posterior for $a_0$ as well as the posterior of $a_0$ with the uniform prior for various $d_{\textrm{MTD}}$ and $d_{\textrm{obs}}$ values.
We observe that when the datasets are perfectly compatible, i.e., $d_{\textrm{obs}}=0$, the marginal posterior of $a_0$ with the optimal prior concentrates around one when $d_{\textrm{MTD}}$ is relatively large.
When $d_{\textrm{obs}}$ increases to $1$, the marginal posterior of $a_0$ concentrates around zero.
The optimal prior becomes increasingly concentrated near one as $d_{\textrm{MTD}}$ increases.
Compared to the marginal posterior with the uniform prior, the optimal prior for $a_0$ achieves a marginal posterior that closely mimics the target distribution when the datasets are fully compatible, while remaining responsive to conflict in the data.


\subsection{Mean Squared Error Criterion}\label{sec4}

In this section, we derive the optimal prior for $a_0$ based on minimizing the MSE. This prior is optimal in the sense that it minimizes the weighted average of the MSEs of the posterior mean of the parameter of interest when its hypothetical true value is equal to its estimate using the historical data, or when it differs from its estimate by the maximum tolerable amount.
Suppose $y_1,\dots,y_n$ and $y_{01},\dots,y_{0{n_0}}$ are observations from a distribution with mean parameter $\mu$. Let $\mu^*$ denote the true value of $\mu$. Let $\bar{\mu}$ denote the posterior mean of $\mu$ using the normalized power prior.
Then, the MSE of $\bar{\mu}$ is 
\begin{align*}
\text{MSE}(\mu^*)=&\int\left[\bar{\mu}(y) - \mu^*\right]^2p(y|\mu^*)dy.
\end{align*}

In the regression setting, $\mu$ is replaced by the regression coefficients $\beta$.

Let $\bar{y}_0$ denote the mean of the historical data.
We aim to find the hyperparameters, $\alpha_0$ and $\beta_0$, for the beta prior for $a_0$ that will minimize
\begin{equation*}
    w\text{MSE}(\mu^*=\bar{y}_0)+(1-w)\text{MSE}(\mu^*=\bar{y}_0+d_{\textrm{MTD}}),
\end{equation*}
where $d_{\textrm{MTD}}$ is the maximum tolerable difference.
Again, we use $d_{\textrm{obs}}=\bar{y}_{\textrm{obs}}-\bar{y}_0$ to represent the difference between the means of the observed current data and the historical data.

\subsubsection{Normal \textit{i.i.d.} Case}\label{sec4.1} 
We demonstrate the use of this criterion for the normal \textit{i.i.d.} case.
Suppose $y_1,\dots,y_n$ and  $y_{01},\dots,y_{0{n_0}}$ are \textit{i.i.d.} observations from N($\mu$, $\sigma^2$) where $\sigma^2=1$ and $n=n_0=30$.
In this example, we fix $\mu^*$ and $y_{MTD}$ at $1.5$, and define $d_{\textrm{MTD}}=\bar{y}_0 -\bar{y}$ and $d_{\textrm{obs}}=\bar{y}_0 -\bar{y}$.
The posterior mean of $\mu$ is computed using Monte Carlo integration and optimization is performed using a grid search.
The optimal prior, optimal posterior, and the posterior using the uniform prior for $a_0$ are plotted in Figure \ref{sim_mse}.
When $d_{\textrm{MTD}}=0.5$, the optimal prior is unimodal with mode around $0.3$.
When $d_{\textrm{MTD}}=1$, the optimal prior is concentrated near zero. When $d_{\textrm{MTD}}=1.5$, the optimal prior is U-shaped and favouring either strong or weak borrowing. When $d_{\textrm{MTD}}$ is small, the algorithm cannot distinguish between the two competing scenarios in the objective function and the resulting optimal prior concentrates around $0.5$.
When $d_{\textrm{MTD}}$ is large, the optimal prior will favour the two scenarios equally.
For columns 2-4, $d_{\textrm{obs}}$ is chosen to be $0$, $1$ and $1.5$. The marginal posterior using the optimal prior concentrates more around zero as $d_{\textrm{obs}}$ increases for a given $d_{\textrm{MTD}}$.
Comparing Figures \ref{sim_mse} and \ref{sim_normal}, we observe that the optimal prior derived using the MSE criterion is more conservative in the sense that it tends to discourage borrowing than that derived using the KL criterion. Note that one might obtain very different optimal priors using the KL criterion versus the MSE criterion, since these metrics have different objectives.

\begin{figure}
\begin{center}
\includegraphics[width=13cm]{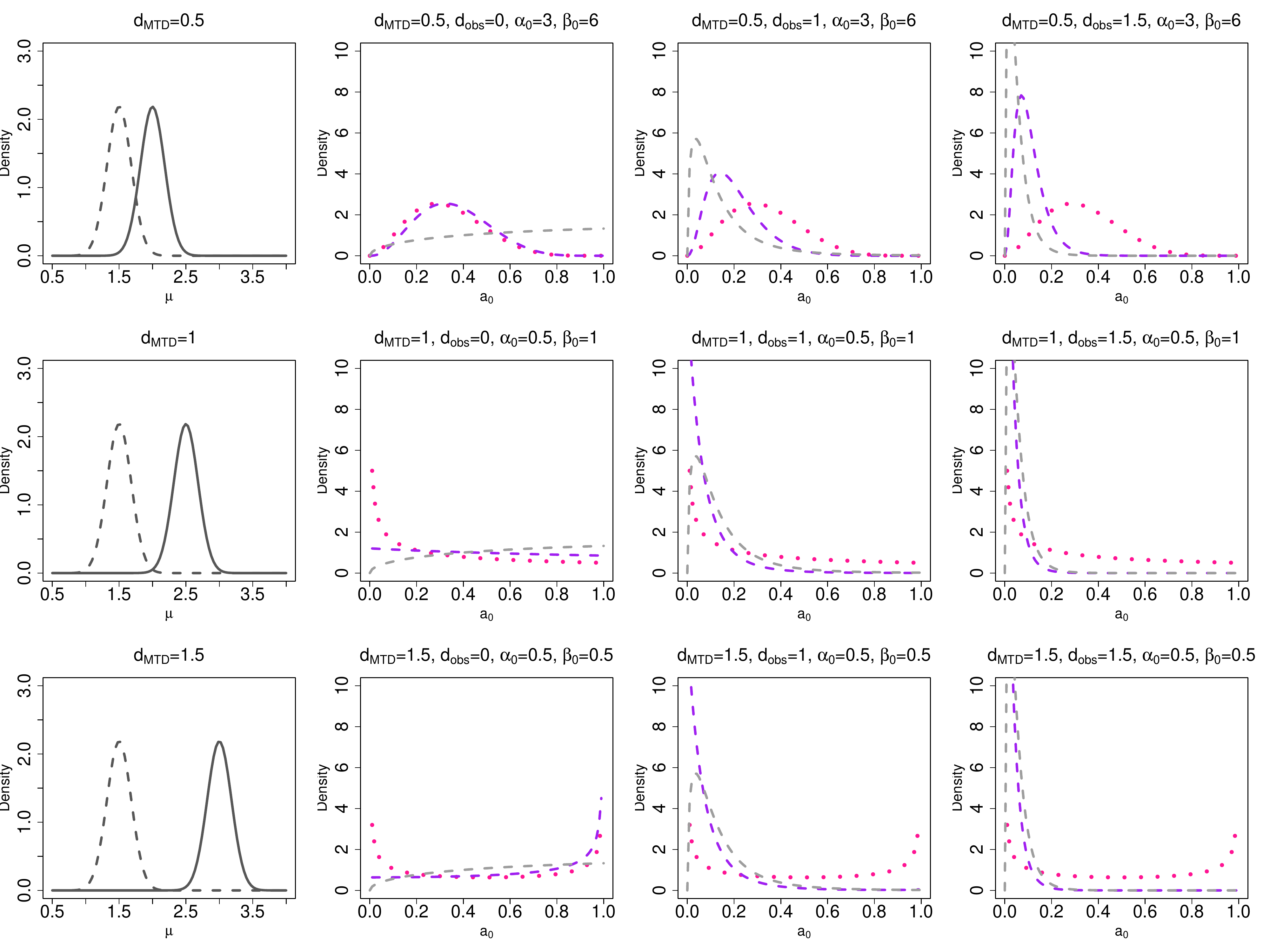}
\end{center}
\caption{Simulation results for the normal \textit{i.i.d.} case when minimizing a convex combination of MSEs when $n=n_0=30$.
The first figure of each row shows the historical (black solid line) and current (black dashed line) data likelihoods if the hypothetical degree of conflict is equal to $d_{\textrm{MTD}}$.
The mean of the hypothetical current data is fixed at $1.5$.
For each row of the figure, the maximum tolerable difference $d_{\textrm{MTD}}$ is chosen to be $0.5$, $1$ and $1.5$, and the corresponding optimal prior (pink dotted line) is derived for each value of $d_{\textrm{MTD}}$.
For each optimal prior, we vary $d_{\textrm{obs}}=\bar{y}_0-\bar{y}_{\textrm{obs}}$. For columns 2-4, $d_{\textrm{obs}}$ is chosen to be $0$, $1$ and $1.5$, respectively.
The purple dashed line represents the marginal posterior of $a_0$ with the optimal prior for a given $d_{\textrm{obs}}$.
The grey dashed line plots the marginal posterior of $a_0$ with the uniform prior.}
\label{sim_mse}
\end{figure}

Table \ref{mse1} shows the MSE for the optimal prior, $\operatorname{beta}(1,1)$ (uniform) and $\operatorname{beta}(2,2)$ as well as the percent reduction of MSE of the optimal prior compared to the uniform prior.
We can see that the percent reduction of MSE increases as $d_{\textrm{MTD}}$ increases.
The table in Appendix \ref{app_bias} displays the decomposition of MSE into bias squared and variance for the three choices of priors.
When $d_{\textrm{MTD}}=0.5$ or $1$, the prior discourages borrowing which results in smaller bias and larger variance.
When $d_{\textrm{MTD}}=1.5$, the model can distinguish easily between the two contrasting objectives, leading to smaller bias and smaller variance. In Appendix \ref{app_rmap}, we compare the MSE of the posterior mean based on  the normalized power prior using the optimal prior for $a_0$ and the robust mixture prior \citep{Schmidli_2014} for the case where the data are i.i.d. normal. A key contribution of this paper is to identify the optimal single beta prior for $a_0$ in an NPP. The comparison in Appendix \ref{app_rmap} is designed to help the reader put this notion of optimality in context by comparing the performance of the posterior mean point estimator from several families of priors. Additional simulation results for varying choices of $n$ and $n_0$ are provided in the Appendix \ref{app_mse}.

\begin{table}
\begin{center}
\caption{MSE for different prior choices and percent reduction of MSE of the optimal prior compared to the uniform prior}\label{tab_mse_1}
\begin{tabular}{lcccc}
 \\
& Optimal Prior & beta$(1,1)$ & beta$(2,2)$ & Percent Reduction of MSE,\\
&&&& Optimal Prior vs. $\operatorname{beta}(1,1)$ \\[5pt]
\hline
$d_{\textrm{MTD}}=0.5$ & 0.054 & 0.057 & 0.057 & 5\% \\
$d_{\textrm{MTD}}=1$ & 0.063 & 0.069 & 0.079  & 9\% \\
$d_{\textrm{MTD}}=1.5$ & 0.052 & 0.059 & 0.067  & 12\% \\
\end{tabular}
\label{mse1}
\end{center}
\end{table}

\section{Case Studies}\label{sec5}

We now illustrate the proposed methodologies by analysing two clinical trial case studies.
First, we study an important application in a pediatric trial where historical data on adults is available. This constitutes a situation of increased importance due to the difficulty in enrolling pediatric patients in clinical trials \citep{fda_guide}. 
Then, we study a classical problem in the analysis clinical trials: using information from a previous study. 
This is illustrated with data on trials of interferon treatment for melanoma.

\subsection{Pediatric Lupus Trial}\label{sec:ped}

Enrolling patients for pediatric trials is often difficult due to the small number of available patients, parental concern regarding safety and technical limitations \citep{psioda_ped}. For many pediatric trials, additional information must be incorporated for any possibility of establishing efficacy \citep{psioda_ped}. The use of Bayesian methods is natural for  extrapolating adult data in pediatric trials through the use of informative priors, and is demonstrated in FDA guidance on complex innovative designs \citep{fda}. 

Belimumab (Benlysta) is a biologic for the treatment of adults with active, autoantibody-positive systemic lupus erythematosus (SLE).
It was proposed that the indication for Belimumab can be expanded to include the treatment of children \citep{psioda_ped}.
The clinical trial PLUTO \citep{ped_2020} has been conducted to examine the effect of Belimumab on children 5 to 17 years of age with active, seropositive SLE who are receiving standard therapy. The PLUTO study has a small sample size due to the rarity of childhood-onset SLE. There have been two previous phase 3 trials, BLISS-52 and BLISS-76 \citep{Furie_2011, Navarra_2011}, which established efficacy of belimumab plus standard therapy for adults. The FDA review of the PLUTO trial submission used data from the adult trials to inform the approval decision \citep{psioda_ped}. All three trials employ the same composite primary outcome, the SLE Responder Index (SRI-4).

We conduct a Bayesian analysis of the PLUTO study incorporating information from the adult studies, BLISS-52 and BLISS-76, using a normalized power prior.
We derive the optimal priors on $a_0$ based on the KL criterion and the MSE criterion.

Our parameter of interest is the treatment effect of Belimumab for children, denoted by $\beta$.
The total sample size of the pooled adult data (BLISS-52 and BLISS-76) is $n_0=1125$ and the treatment effect is $0.481$.
We choose $d_{\textrm{MTD}}=0.481$ which equals the treatment effect of the historical study.
The pediatric data has a sample size of $92$ and the estimated treatment effect is $0.371$.
We use the asymptotic normal approximation to the logistic regression model (see equation \eqref{glm_marg_a0} in Appendix \ref{glm_proof}) 
with one covariate (the treatment indicator).
We choose $w=0.5$ and $n=100$ (sample size of the simulated current dataset).
For the KL criterion, the objective function $K$ is computed using Monte Carlo integration and optimization is performed using the \verb|optim()| function in R. For the MSE criterion, the posterior mean of $\beta$ is computed using the R package \textit{BayesPPD} and optimization is performed using a grid search where the values of $\alpha_0$ and $\beta_0$ range from $0.5$ to $6$ with an increment of $0.5$.
The optimal priors derived using the KL criterion and MSE criterion are displayed in Figure \ref{real_ped}.
Table \ref{tab:ped_beta} (left) provides the posterior mean, standard deviation and 95\% credible interval for $\beta$ using the optimal priors and several other beta priors for comparison.
We observe that the optimal prior derived using the KL criterion leads to a lower posterior standard deviation compared to the uniform prior because more historical information is borrowed. Note that when a $\operatorname{beta}(1, 10)$ prior is used for $a_0$, i.e., very little historical information is borrowed, the $95\%$ credible interval of $\beta$ includes zero. In Appendix \ref{app_design}, we demonstrate using the proposed optimal priors in a clinical trial design application for the pediatric lupus trial. 

\begin{figure}
\centering
\includegraphics[width=13cm]{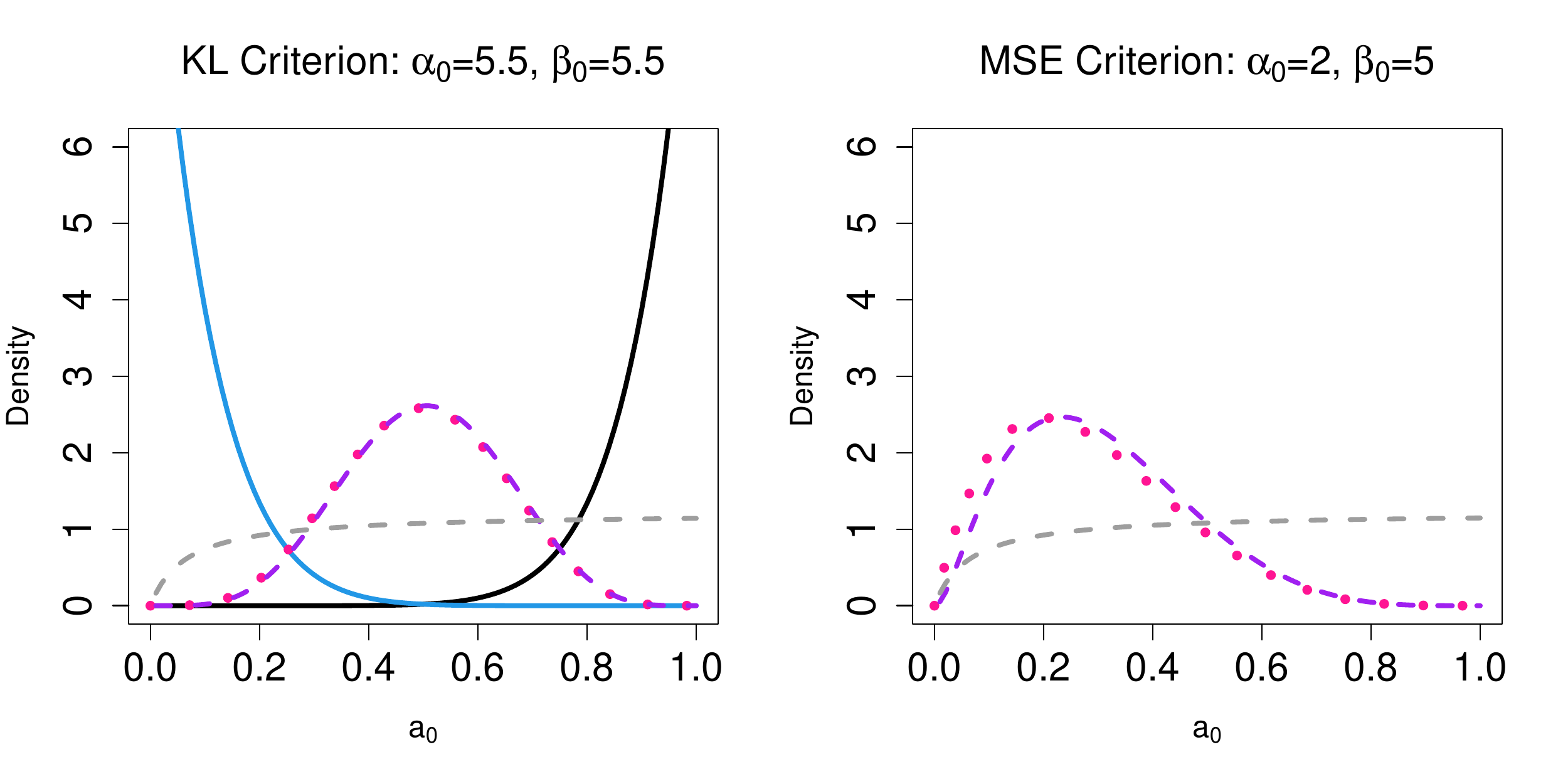}
\caption{After combining studies BLISS-52 and BLISS-76 for adults, the total sample size is $n_0=1125$ and log odds ratio for treatment vs. control group is $0.481$.
We choose $d_{\textrm{MTD}}=0.481$ to be the maximum tolerable difference.
The pediatric data has a sample size of $n=92$.
The actual observed log odds ratio is $0.371$.
The figure on the left displays the optimal prior (pink dotted line) and posterior (purple dashed line) derived using the KL criterion.
The figure on the right displays the optimal prior for $a_0$ and the posterior derived using the MSE criterion.
The posterior of $a_0$ using the uniform prior (grey dashed line) is also shown.}
\label{real_ped}
\end{figure}

\begin{table}
\caption{Pediatric lupus trial (left) and melanoma trial (right): posterior mean, standard deviation, and 95\% credible interval for $\beta$}
\begin{center}
\scalebox{0.8}{
\begin{tabular}{lccclccc}
&Lupus &Trial&&&Melanoma &Trial\\
\hline
Prior for $a_0$ & Mean & SD & 95\% CI & Prior for $a_0$ &Mean & SD & 95\% CI  \\[5pt]
\hline
$\operatorname{beta}(5.5, 5.5)$ & 0.47 & 0.16 & (0.15, 0.79) & $\operatorname{beta}(0.7, 1.5)$ & 0.05 & 0.19 & (-0.33, 0.43)  \\
(optimal by KL)&&& & (optimal by KL)&&&\\
$\operatorname{beta}(2, 5)$ & 0.47 & 0.21 & (0.04, 0.89)  & $\operatorname{beta}(5.5, 3)$ & -0.01 & 0.17 & (-0.35, 0.33)   \\
(optimal by MSE)&&&&(optimal by MSE)&&&\\
$\operatorname{beta}(1, 1)$& 0.47 & 0.18 & (0.12, 0.83) &$\operatorname{beta}(1, 1)$& 0.04 & 0.19 & (-0.32, 0.42)  \\
$\operatorname{beta}(2, 2)$ & 0.47 & 0.17 & (0.13, 0.79) & $\operatorname{beta}(2, 2)$ & 0.03 & 0.18 & (-0.34, 0.40) \\
$\operatorname{beta}(0.5, 0.5)$ & 0.47 & 0.17 & (0.12, 0.81) & $\operatorname{beta}(0.5, 0.5)$ & 0.05 & 0.19 & (-0.33, 0.41)\\
$\operatorname{beta}(10, 1)$ & 0.48 & 0.12 & (0.24, 0.71) & $\operatorname{beta}(10, 1)$ & -0.08 & 0.16 & (-0.38, 0.23)\\
$\operatorname{beta}(1, 10)$ & 0.44 & 0.28 & (-0.13, 1) & $\operatorname{beta}(1, 10)$ & 0.06 & 0.19 & (-0.3, 0.44)\\
\end{tabular}
}
\end{center}
\label{tab:ped_beta}
\end{table}

\subsection{Melanoma Trial}\label{sec:mel}

Interferon Alpha-2b (IFN) is an adjuvant chemotherapy for deep primary or regionally metastatic melanoma.
IFN was used in two phase 3 randomized controlled clinical trials, E1684 and E1690 \citep{mel_1996}.
In this example, we choose overall survival (indicator for death) as the primary outcome.
We conduct a Bayesian analysis of the E1690 trial incorporating information from the E1684 trial, using a normalized power prior.
We include three covariates in the analysis, the treatment indicator, sex and the logarithm of age.
As before, we obtain the optimal priors for $a_0$ based on both the KL criterion and the MSE criterion.

Our parameter of interest is the treatment effect of IFN, denoted by $\beta$.
The total sample size of the E1684 trial is $n_0=285$ and the treatment effect is $-0.423$.
We choose $d_{\textrm{MTD}}=0.423$ which equals the treatment effect of the historical study.
The E1690 trial has a sample size of $427$ and the treatment effect is $0.098$.
We use the asymptotic normal approximation to the  logistic regression model (see equation \eqref{glm_marg_a0} in Appendix \ref{glm_proof}) 
with three covariates.
We choose $w=0.5$, $d_{\textrm{MTD}}=0.423$ and $n=400$ (sample size of the simulated current dataset).
For the KL criterion, the objective function $K$ is computed using Monte Carlo integration and optimization is performed using the \verb|optim()| function in R.
For the MSE criterion, the posterior mean of $\beta$ is computed using the R package \textit{BayesPPD} and optimization is performed using a grid search where the values of $\alpha_0$ and $\beta_0$ range from $0.5$ to $6$ with an increment of $0.5$.
The optimal priors derived using the KL criterion and MSE criterion are displayed in Figure \ref{real_mel}.
The optimal prior derived using the KL criterion is $\operatorname{beta}(0.7, 1.5)$, which has density around zero.
For the MSE criterion, the optimal prior derived is $\operatorname{beta}(5.5, 3)$, which is unimodal with mode around $0.7$.
This is likely due to the fact that $d_{\textrm{MTD}}$ is small relative to the total sample size of $712$ -- see  also simulations in Appendix \ref{app_mse}.
Because the observed difference is larger than $d_{\textrm{MTD}}$, the marginal posterior of $a_0$ has mode around $0.4$, which discourages more strongly than the prior.
Table \ref{tab:ped_beta} (right) provides the posterior mean, standard deviation and 95\% credible interval for $\beta$ using the optimal priors and several other beta priors for comparison.
Compared to the uniform prior, the optimal prior derived using the KL criterion results in a larger posterior mean, indicating that less historical information is borrowed. Compared to the uniform prior, the optimal prior derived using the MSE criterion borrows more historical information, resulting in a smaller posterior mean and a smaller variance.

\begin{figure}
\centering
\includegraphics[width=13cm]{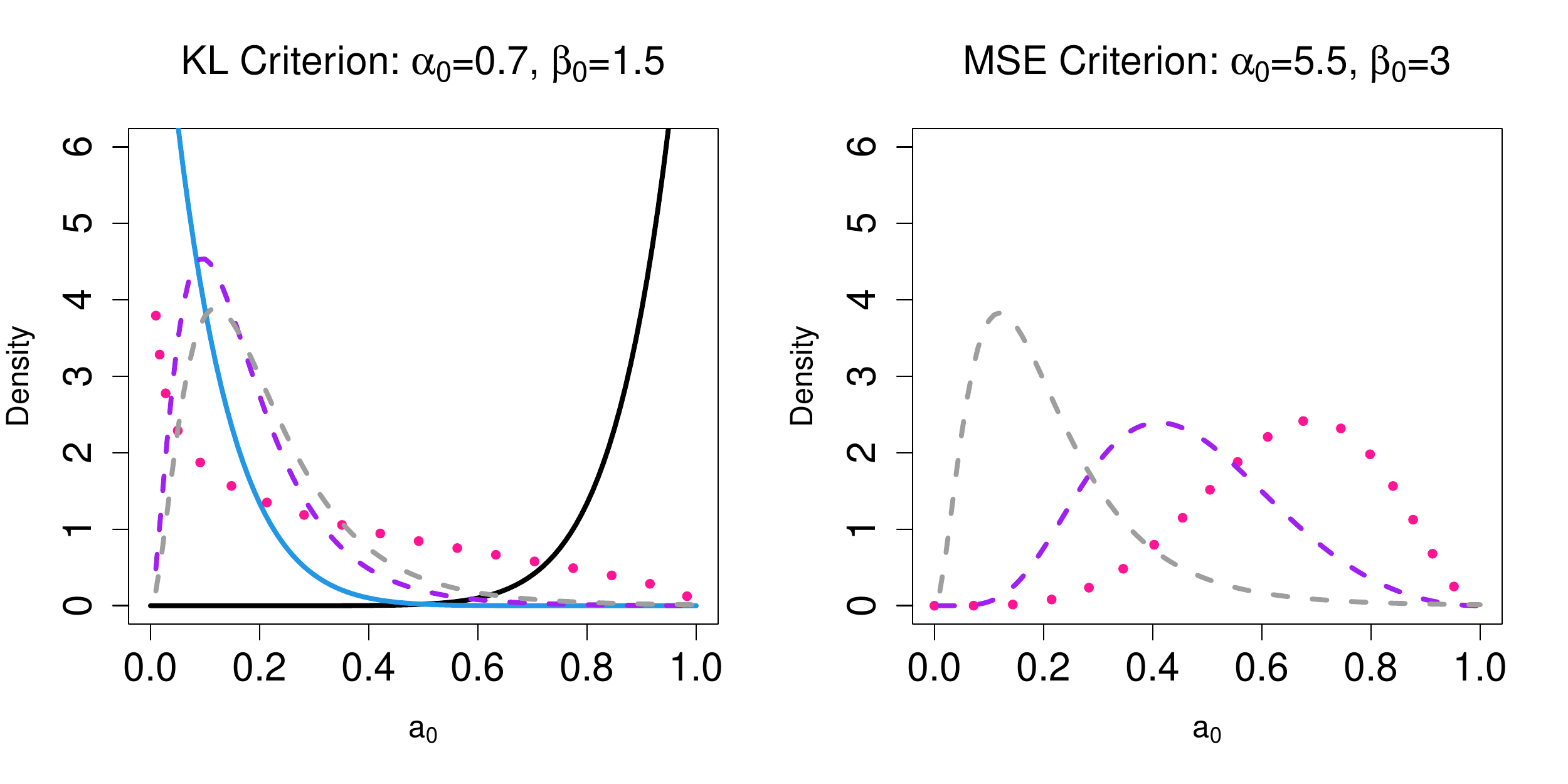}
\caption{The total sample size of the E1684 trial is $n_0=285$ and log odds ratio for treatment vs. control group is $-0.423$.
We choose $d_{\textrm{MTD}}=0.423$ to be the maximum tolerable difference. The E1690 trial has sample size $n=427$.
The observed log odds ratio is $0.098$.
The figure on the left displays the optimal prior (pink dotted line) and posterior (purple dashed line) derived using the KL criterion.
The figure on the right displays the optimal prior for $a_0$ and the posterior derived using the MSE criterion.
The posterior of $a_0$ using the uniform prior (grey dashed line) is also shown.}
\label{real_mel}
\end{figure}




\section{Discussion}

In this paper, we have explored the asymptotic properties of the normalized power prior when the historical and current data are compatible and when they are incompatible.
Our results demonstrate that there is a fundamental asymmetry in the adaptive borrowing properties of the normalized power prior: while for any discrepancy between the historical and current data sets one can expect the marginal posterior of the discounting parameter $a_0$ to concentrate around zero, under complete compatibility this distribution does not \textit{concentrate} around $1$ -- although the mode is at $1$ under mild conditions~\citep{Han_Ye_Wang_2022}.
This has been seen as flaw with the NPP by some authors (e.g.~\cite{Pawel_2023}), but we argue this view is misguided: expecting the posterior of $a_0$ to concentrate around $1$ is effectively expecting two data sets to support exchangeability.
This is a notoriously hard task, as exemplified by the difficulty in estimating the population variance in hierarchical models. Moreover, we have shown that  for an i.i.d. normal model with finite samples, the marginal posterior of $a_0$ always has more mass around one when the datasets are fully compatible, compared to the case where there is any discrepancy. 

We have proposed two criteria based on which the optimal hyperparameters of the prior for $a_0$ can be derived.
While the exact values of the hyperparameters can be obtained using our objective functions, we suggest the following rules of thumb for estimating the optimal prior given different choices of the maximum tolerable difference.
When the KL criterion is used, a beta distribution centered around $0.5$, such as the $\operatorname{beta}(2, 2)$, is optimal for small values (when plots of the current and historical data likelihoods substantially overlap) of maximum tolerable difference , while a beta distribution with mean close to $1$, such as the $\operatorname{beta}(2, 0.5)$, should be used for large values of maximum tolerable difference.
When the MSE criterion is used, a beta distribution with mean less than $0.5$, such as the $\operatorname{beta}(3, 6)$, is optimal for small values of maximum tolerable difference, while a beta distribution with modes at zero and one, as for example a $\operatorname{beta}(0.5, 0.5)$, should be used for large values of maximum tolerable difference.
The MSE criterion is a more conservative criterion, in the sense that it tends to discourage borrowing, than the KL criterion.

We observe that in Figures 1-3, the marginal posterior of $a_0$ is bimodal, with modes at zero and one, in some cases. Future work should further examine the issue of multi-modality. Other potential future work includes extending our method to survival and longitudinal outcomes, as well as accommodating dependent discounting parameters when multiple historical datasets are available.

\begin{appendix}
\label{sec6}

\newpage

\section{Proofs from Section 2}

\subsection{Technical conditions for the limit theorems}

We start our presentation by stating technical conditions under which the limiting theorems presented in Section 2 hold.
Then, we state an important result below (Bayes Central Limit Theorem (\cite{chen_1985})) which gives support to many of the proofs herein.
In what follows, we will follow \cite{chen_1985} in establishing the necessary conditions for the limiting posterior density to be normal. 
Let the parameter space of interest be $\Theta$ and a $p$-dimensional Euclidean space and let $B_r(a) = \{\theta \in \Theta : |\theta-a|\leq r\}$ be a neighbourhood of size $r$ of the point $a \in \Theta$.
Also, write $L_n(\theta) := \sum_{i=1}^n \log f(x\mid \theta)$.\\

\begin{thm}[Bayes Central Limit Theorem (\cite{chen_1985})]
Suppose that for each $n > N$ with $N >0$, $L_n$  attains a strict local maximum $\hat{\theta}_n$ such that $L_n^{\prime}(\hat{\theta}_n) := \frac{\partial}{\partial \theta}L_n(\hat{\theta}_n) = 0$ and the Hessian $L_n^{\prime\prime}(\theta) := \frac{\partial^2}{\partial \theta^2}L_n(\theta)$ is negative-definite for all $\theta \in \Theta$.

Moreover, suppose $\hat{\theta}_n$ converges almost surely to $\theta_0 \in \Theta$  as $n \to \infty$ and the prior density $\pi(\theta)$ is positive and continuous at $\theta_0$.
Assume that the following conditions hold:
\begin{enumerate}
    \item[C1] The largest eigenvalue of $\left[- L_n^{\prime\prime}(\hat{\theta}_n) \right]^{-1} \to 0$ a.s. as $n \to \infty$;
    \item[C2] For $\varepsilon > 0$ there exists (a.s.) $N_\varepsilon >0$ and $r > 0$ such that for all $n > \max\left\{N, N_\varepsilon \right\}$ and $\theta \in B_r(\hat{\theta}_n)$, $L_n^{\prime\prime}(\theta)$ is well-defined and 
    \begin{equation*}
        I_p - A(\varepsilon) \leq L_n^{\prime\prime}(\theta) \left[L_n^{\prime\prime}(\hat{\theta}_n) \right]^{-1} \leq I_p + A(\varepsilon),
    \end{equation*}
    where $I_p$ is the $p$-dimensional identity matrix and $A(\varepsilon)$ is a $p \times p$ positive semidefinite matrix whose largest eigenvalue goes to zero as $\varepsilon \to 0$.
    \item[C3] The sequence of posterior distributions $p_n(\theta \mid x)$ satisfies, as $n \to \infty$,
    \begin{equation*}
        \int_{\Theta \setminus B_r(\hat{\theta}_n)}p_n(t \mid x)\,dt \to 0, a.s.,
    \end{equation*}
    for $r >0$, i.e., the sequence of posteriors is \emph{consistent} at $\hat{\theta}_n$.
    Here we have assumed that the support of the posterior distributions is $\Theta$, but this could be replaced by a sequence $\Theta_n$.
\end{enumerate}
Then we say that the posteriors converge in distribution to a normal with parameters $\hat{\theta}_n$ and $ \left[-L_n^{\prime\prime}(\hat{\theta}_n)\right]^{-1}$.
\end{thm}
For notational convenience we will (somewhat informally) write
$$
p_n(\theta|x) \rightarrow N_p\left(\hat{\theta}_n, \left[-L_n^{\prime\prime}(\hat{\theta}_n)\right]^{-1}\right),
$$
as $n \rightarrow \infty$.
This should be understood as the posterior density becoming highly peaked and behaving like a normal kernel around $\hat{\theta}_n)$ \cite[page 541]{chen_1985}.
Since the probability outside $B_r(\hat{\theta}_n)$ is negligible, one needs not to concern oneself with what happens on $\Theta \setminus B_r(\hat{\theta}_n)$ when taking posterior expectations, for instance.
See also Theorem 7.89 in \cite{Schervish1995} (page 437).

\subsection{Proof of Theorem \ref{nocov_th}}\label{nocov_proof}

Now we move on to present a proof for Theorem 2.1 in Section 2, which discusses the concentration of the posterior of $a_0$ at zero as the sample sizes increase in the case when there is some discrepancy between the historical and current data sets.

\begin{proof}
We first employ the Bayes Central Limit Theorem presented above to rewrite the limiting marginal posterior distribution of $a_0$. 
Under the regularity conditions as $n \rightarrow \infty$,
\begin{align*}
L_n(\theta|D) &\rightarrow N(\hat{\theta}_n, v_n), \quad \textrm{and}\\
\frac{1}{c(a_0)}L_{n_0}(\theta|D_0)^{a_0} \pi_0(\theta) &\rightarrow N(\hat{\theta}_0, v_0(a_0)),
\end{align*}
where $\hat{\theta}_n=\dot{b}^{-1}(\bar{y})$, $\hat{\theta}_0=\dot{b}^{-1}(\bar{y}_0)$, $v_n=(n\ddot{b}(\hat{\theta}_n))^{-1}$, and $v_0(a_0)=(a_0n_0\ddot{b}(\hat{\theta}_0))^{-1}$.
For simplicity of notation, let $v_0=v_0(a_0)$, $\ddot{b}^{-1}= \ddot{b}^{-1}(\hat{\theta}_n)$ and $\ddot{b}_0^{-1}=\ddot{b}^{-1}(\hat{\theta}_0)$.
Then the kernel of the marginal posterior of $a_0$ becomes
\begin{align*}
\pi^*(a_0|D_0, D,\alpha_0,\beta_0)
\equiv& \int L_n(\theta|D)\frac{L_{n_0}(\theta|D_0)^{a_0} \pi_0(\theta)}{c(a_0)}a_0^{\alpha_0-1}(1-a_0)^{\beta_0-1}\,d\theta, \\
\rightarrow& a_0^{\alpha_0-1}(1-a_0)^{\beta_0-1}\int N(\hat{\theta}_n, v_n)N(\hat{\theta}_0, v_0(a_0))d\theta,\\
\propto& a_0^{\alpha_0-1}(1-a_0)^{\beta_0-1}v_0^{-\frac{1}{2}}\int \exp\left\{-\frac{1}{2v_n}(\theta-\hat{\theta}_n)^2\right\}\exp\left\{-\frac{1}{2v_0}(\theta-\hat{\theta}_0)^2\right\}d\theta,\\
\propto &a_0^{\alpha_0-1}(1-a_0)^{\beta_0-1}\left(\frac{v_n+v_0}{v_n}\right)^{-\frac{1}{2}}\exp\left\{-\frac{1}{2}\left[\frac{v_n\hat{\theta}_0^2-v_0\hat{\theta}_n^2-2v_n\hat{\theta}_n\hat{\theta}_0}{(v_0+v_n)v_n}\right]\right\},\\
= &a_0^{\alpha_0-1}(1-a_0)^{\beta_0-1}\left(\frac{v+v_0}{v_n}\right)^{-\frac{1}{2}}\exp\left\{\frac{v_0\hat{\theta}_n^2-v_n(\delta^2-\hat{\theta}_n^2)}{2(v_0+v_n)v_n}\right\}\: \textrm{(since } |\hat{\theta}_n-\hat{\theta}_0|=\delta),\\
= & a_0^{\alpha_0-1}(1-a_0)^{\beta_0-1}\left(\frac{v_n+v_0}{v_n}\right)^{-\frac{1}{2}}\exp\left\{\frac{\hat{\theta}_n^2}{2v}-\frac{\delta^2}{2(v+v_0)}\right\},\\
= &a_0^{\alpha_0-1}(1-a_0)^{\beta_0-1}\left(\frac{v_n+v_0}{v_n}\right)^{-\frac{1}{2}}\exp\left\{\frac{\hat{\theta}_n^2}{2v_n}-\frac{na_0r\delta^2}{2(\ddot{b}_0^{-1}+a_0r\ddot{b}^{-1})}\right\}.\\
\end{align*}
Then the marginal posterior of $a_0$ becomes 
\begin{align}
\pi(a_0|D_0, D,\alpha_0,\beta_0) =& \frac{\pi^*(a_0|D_0, D,\alpha_0,\beta_0)}{\int \pi^*(a_0|D_0, D,\alpha_0,\beta_0)\, da_0},\\
\rightarrow& \frac{a_0^{\alpha_0-1}(1-a_0)^{\beta_0-1}[\ddot{b}^{-1}+(a_0r\ddot{b}_0)^{-1}]^{-\frac{1}{2}}\exp\left\{-\frac{na_0r\delta^2}{2(\ddot{b}_0^{-1}+a_0r\ddot{b}^{-1})}\right\}}{\int a_0^{\alpha_0-1}(1-a_0)^{\beta_0-1}[\ddot{b}^{-1}+(a_0r\ddot{b}_0)^{-1}]^{-\frac{1}{2}}\exp\left\{-\frac{na_0r\delta^2}{2(\ddot{b}_0^{-1}+a_0r\ddot{b}^{-1})}\right\}\,da_0},\\
=& \frac{a_0^{\alpha_0-1}(1-a_0)^{\beta_0-1}\left[\frac{a_0r}{1+a_0r\frac{\ddot{b}^{-1}}{\ddot{b}_0^{-1}}}\right]^{\frac{1}{2}}\exp\left\{-\frac{na_0r\delta^2}{2\ddot{b}_0^{-1}\left(1+a_0r\frac{\ddot{b}^{-1}}{\ddot{b}_0^{-1}}\right)}\right\}}{\int a_0^{\alpha_0-1}(1-a_0)^{\beta_0-1}\left[\frac{a_0r}{1+a_0r\frac{\ddot{b}^{-1}}{\ddot{b}_0^{-1}}}\right]^{\frac{1}{2}}\exp\left\{-\frac{na_0r\delta^2}{2\ddot{b}_0^{-1}\left(1+a_0r\frac{\ddot{b}^{-1}}{\ddot{b}_0^{-1}}\right)}\right\}da_0}.\label{marg_a0}
\end{align}


Let $h(a_0)=\frac{a_0r\delta^2}{2\ddot{b}_0^{-1}\left(1+a_0r\frac{\ddot{b}^{-1}}{\ddot{b}_0^{-1}}\right)}$ and $f(a_0)=\left[\frac{a_0r}{1+a_0r\frac{\ddot{b}^{-1}}{\ddot{b}_0^{-1}}}\right]^{\frac{1}{2}}$.
Then the denominator is $$A = \int_0^1 a_{0}^{\alpha_0-1}(1-a_{0})^{\beta_0-1}f(a_0)\exp\left\{-nh(a_0)\right\}da_0.$$ Let $A=A_1+A_2$ where $$A_1 = \int_0^{\epsilon}a_{0}^{\alpha_0-1}(1-a_{0})^{\beta_0-1}f(a_0)\exp\left\{-nh(a_0)\right\}da_0$$ and $$A_2 = \int_{\epsilon}^1a_{0}^{\alpha_0-1}(1-a_{0})^{\beta_0-1}f(a_0)\exp\left\{-nh(a_0)\right\}da_0.$$
We want to show $\lim\limits_{n\rightarrow \infty} \frac{A_2}{A_1} = 0$. \\
First, we can see that
$$h'(a_0)=\frac{r\delta^2}{2\ddot{b}_0^{-1}}\left(\frac{a_0}{1+a_0r\frac{\ddot{b}^{-1}}{\ddot{b}_0^{-1}}}\right)'=\frac{r\delta^2}{2\ddot{b}_0^{-1}}\left(\frac{1+a_0r\frac{\ddot{b}^{-1}}{\ddot{b}_0^{-1}}-a_0r\frac{\ddot{b}^{-1}}{\ddot{b}_0^{-1}}}{(1+a_0r\frac{\ddot{b}^{-1}}{\ddot{b}_0^{-1}})^2}\right)=\frac{r\delta^2}{2\ddot{b}_0^{-1}}\left(1+a_0r\frac{\ddot{b}^{-1}}{\ddot{b}_0^{-1}}\right)^{-2} > 0.$$
Then $\inf_{x \in [\epsilon,1]} h(x)=h(\epsilon).$ We can also see that $h'(a_0)$ is continuous since $1+a_0r\frac{\ddot{b}^{-1}}{\ddot{b}_0^{-1}}$ is nonzero on $(0,1)$.\\
We then observe that $$f'(a_0)=\frac{1}{2}\left[\frac{a_0r}{1+a_0r\frac{\ddot{b}^{-1}}{\ddot{b}_0^{-1}}}\right]^{-\frac{1}{2}}\frac{r}{(1+a_0r\frac{\ddot{b}^{-1}}{\ddot{b}_0^{-1}})^2} > 0.$$
Thus $\sup_{ x \in [\epsilon,1]} f(x) = f(1)$.\\ 
Now we are ready to find the upper bound of $A_2$.
Since, for any $a_0 \in [\epsilon, 1]$, $f(a_0) \le f(1)$ and $\exp(-nh(a_0)) \le \exp(-nh(\epsilon))$ , we have 
\begin{align*}
A_2 &\le f(1)\exp(-nh(\epsilon))\int_{\epsilon}^{1}a_0^{\alpha_0-1}(1-a_0)^{\beta_0-1}da_0\\
&\le f(1)\exp(-nh(\epsilon))\int_0^{1}a_0^{\alpha_0-1}(1-a_0)^{\beta_0-1}da_0,\\
&=f(1)\exp(-nh(\epsilon))\frac{\Gamma(\alpha_0)\Gamma(\beta_0)}{\Gamma(\alpha_0+\beta_0)}=C_1\exp(-nh(\epsilon)),
\end{align*}
where $C_1 > 0$ is an integration constant.
Now we find the lower bound of $A_1$.
We know that $$A_1 \ge \int_{\frac{\epsilon}{2}}^{\epsilon}a_{0}^{\alpha_0-1}(1-a_{0})^{\beta_0-1}f(a_0)\exp\left\{-nh(a_0)\right\}da_0.$$
Further, $a_0^{\alpha_0-1} \ge \min((\frac{\epsilon}{2})^{\alpha_0-1}, \epsilon^{\alpha_0-1})$, corresponding to $\alpha_0\ge 1$ and $\alpha_0 < 1$, respectively.
Similarly, $(1-a_0)^{\beta_0-1} \le \min((1-\epsilon)^{\beta_0-1} , (1-\frac{\epsilon}{2})^{\beta_0-1})$, corresponding to $\beta_0\ge 1$ and $\beta_0 < 1$, respectively.
Since $h''(a_0) < 0$, $\sup_{x \in [\frac{\epsilon}{2}, \epsilon]}h'(x)=h'(\frac{\epsilon}{2})$.
In addition, $\inf_{x \in [\frac{\epsilon}{2}, \epsilon]}f(x) = f(\frac{\epsilon}{2})$.
Then we have 
\begin{align*}
A_1 &\ge \int_{\frac{\epsilon}{2}}^{\epsilon}a_{0}^{\alpha_0-1}(1-a_{0})^{\beta_0-1}f(a_0)\frac{1}{h'(a_0)}\exp\left\{-nh(a_0)\right\}h'(a_0)\,da_0,\\
&= \int_{\frac{\epsilon}{2}}^{\epsilon}a_{0}^{\alpha_0-1}(1-a_{0})^{\beta_0-1}f(a_0)\frac{1}{h'(a_0)}\exp\left\{-nh(a_0)\right\}\,dh(a_0),\\
& \ge f(\epsilon/2)\times\min((\epsilon/2)^{\alpha_0-1}, \epsilon^{\alpha_0-1})\times\min((1-\epsilon)^{\beta_0-1} , (1-\epsilon/2)^{\beta_0-1})\times (h'(\epsilon/2))^{-1}\\
&\times\int_{\frac{\epsilon}{2}}^{\epsilon}\exp(-nh(a_0))\,dh(a_0),\\
&=C_2\int_{\frac{\epsilon}{2}}^{\epsilon}\exp(-nh(a_0))\,dh(a_0),\\
&=C_2\frac{1}{n}[\exp(-nh(\epsilon/2))-\exp(-nh(\epsilon))],
\end{align*}
where $C_2 > 0$ is again an integration constant.
Therefore, $$0 \le \frac{A_2}{A_1} \le \frac{C_1\exp(-nh(\epsilon))}{C_2\frac{1}{n}[\exp(-nh(\epsilon/2))-\exp(-nh(\epsilon))]}=\frac{C_1n}{C_2[\exp(-n[h(\epsilon/2)-h(\epsilon)])-1]}.$$ 
Thus, $\lim\limits_{n\rightarrow\infty} \frac{A_2}{A_1} = 0$ by L'Hopital's rule. Since $\frac{A_2}{A_1} \ge \frac{A_2}{A}$, $\lim\limits_{n\rightarrow\infty} \frac{A_2}{A} = 0$.
Hence, $\lim\limits_{n\rightarrow\infty} \frac{A_1}{A} = 1$. 
\end{proof}

\subsection{Proof for Corollary \ref{nocov_coro}}
\label{nocov_coro_proof}
\begin{proof}
The result follows by setting $\delta=0$ and $\ddot{b}^{-1}=\ddot{b}_0^{-1}$ into \eqref{marg_a0}. 
\end{proof}


\subsection{Proof for Theorem \ref{cdf}}\label{cdf_proof}
\begin{proof}
Let $r=\frac{n_0}{n}$. Since $y_1,\dots,y_n$ and $y_{01}, \dots, y_{n_0}$ are i.i.d. normal data, 
the marginal posterior of $a_0$ is 
\begin{align*}
\pi(a_0|D_0, D)
=& \frac{\pi(a_0)\left[\frac{a_0r}{1+a_0r\frac{\sigma^2}{\sigma_0^2}}\right]^{\frac{1}{2}}\exp\left\{-\frac{na_0rd^2}{2\sigma_0^2\left(1+a_0r\frac{\sigma^2}{\sigma_0^2}\right)}\right\}}{\int \pi(a_0)\left[\frac{a_0r}{1+a_0r\frac{\sigma^2}{\sigma_0^2}}\right]^{\frac{1}{2}}\exp\left\{-\frac{na_0rd^2}{2\sigma_0^2\left(1+a_0r\frac{\sigma^2}{\sigma_0^2}\right)}\right\}da_0}.
\end{align*}
With
\begin{equation*}
    g_d(a_0) := \pi(a_0)\left[\frac{a_0r}{1+a_0r\frac{\sigma^2}{\sigma_0^2}}\right]^{\frac{1}{2}}\exp\left\{-\frac{na_0rd^2}{2\sigma_0^2\left(1+a_0r\frac{\sigma^2}{\sigma_0^2}\right)}\right\},
\end{equation*}
we write
\begin{equation*}
    F_d(a_0) := \frac{\int_{0}^{a_0} g_d(x)\,dx}{\int_{0}^1 g_d(x)\,dx} = \frac{G_d(a_0)}{G_d(1)}.
\end{equation*}

We want to show that
\begin{equation}
\label{eq:condition}
    \frac{\partial F_d(a_0) }{\partial d} > 0, a_0 \in (0, 1).
\end{equation}

Using the quotient rule we conclude that \eqref{eq:condition} holds if and only if:
\begin{equation*}
\frac{\partial}{\partial d}G_d(a_0)G_d(1) - \frac{\partial}{\partial d}G_d(1)G_d(a_0) > 0.
\end{equation*}
We note that
\begin{align*}
    \frac{\partial}{\partial d}G_d(a_0) &= - d \frac{nr}{\sigma_0^2} \int_{0}^{a_0} h(x) g_d(x)\,dx, \quad  \text{with}\\
     h(a_0) &= \frac{a_0}{1 + a_0r\frac{\sigma^2}{\sigma_0^2}}.
\end{align*}
This in turn means that \eqref{eq:condition} is equivalent to
\begin{equation*}
    \int_{0}^1 h(x)g_d(x)\,dx\int_{0}^{a_0} g_d(x)\,dx - \int_{0}^{a_0} h(x)g_d(x)\,dx\int_{0}^1 g_d(x)\,dx > 0,
\end{equation*}
i.e.,
\begin{align*}
    \frac{\int_{0}^{a_0} g_d(x)\,dx}{\int_{0}^1 g_d(x)\,dx} &> \frac{\int_{0}^{a_0} h(x)g_d(x)\,dx}{\int_{0}^1 h(x)g_d(x)\,dx}.
\end{align*}
We first prove the following lemma. 
\begin{lem}[Ratio of truncated expectations]
\label{lem:trunc}
Let $X$ be a random variable in $(0,1)$ with distribution function $F$.
Take any positive increasing function $h$.
Then
\begin{equation*}
\label{eq:claim}
    \frac{E[h(X)\mathbb{I}(X \leq a)]}{E[h(X)]} < F(a),
\end{equation*}
for $a \in (0, 1)$.
\end{lem}
\begin{proof}
Start by dividing through by $F(a)$ to get
\begin{equation*}
    \frac{E[h(X)\mid X \leq a]}{E[h(X)]} < 1.
\end{equation*}
But by the law of total expectation, we have
$$E[h(X)] = E[h(X) \mid X \leq a]F(a) + E[h(X) \mid X > a][1-F(a)],$$
thus the LHS is
\begin{equation*}
    \frac{E[h(X)\mid X \leq a]}{E[h(X) \mid X \leq a]F(a) + E[h(X) \mid X > a][1-F(a)]}.
\end{equation*}
If we let $w = E[h(X)\mid X \leq a]$ and $u = E[h(X) \mid X > a]$, we have that $u = w + \varepsilon$ with $\epsilon > 0$.
Putting $\alpha = F(a)$, we have
\begin{align*}
    \frac{w}{\alpha w + (1 - \alpha) u} &= \frac{w}{\alpha (w-u) + u},\\
    &= \frac{w}{w + (1 - \alpha)\varepsilon} < 1,    
\end{align*}
which concludes the argument.
\end{proof}
We may assume without loss of generality that $g_d$ is a normalised density. Since $h(a_0)$ is increasing, we apply Lemma \ref{lem:trunc}, which completes the proof. 
\end{proof}


\subsection{Proof for Theorem \ref{glm_th}}\label{glm_proof}

\begin{proof}
By the Bayes Central Limit Theorem, we know that 
$$L_n(\beta|D) \rightarrow N(\hat{\beta}, \Sigma(\hat{\beta})),$$ where $\Sigma(\beta)=-\left[\frac{\partial^2\log [L_n(\beta|D)]}{\partial\beta_i\partial\beta_j}\right]^{-1}$, and also
$$\frac{1}{c^*(a_0)}L_{n_0}(\beta|D_0)^{a_0} \pi_0(\beta)\rightarrow N(\hat{\beta}_0, \Sigma_0(a_0, \hat{\beta})),$$ where $\Sigma_0(a_0, \beta) = -\left[\frac{\partial^2\log [L_{n_0}(\beta|D_0)^{a_0}\pi_0(\beta)]}{\partial\beta_i\partial\beta_j}\right]^{-1}$. For simplicity of notation, let $\Sigma= \Sigma(\hat{\beta})$ and $\Sigma_0=\Sigma_0(a_0, \hat{\beta})$.
Then the marginal posterior of $a_0$ becomes
\begin{align*}
\pi(a_0|D_0, D, \alpha_0, \beta_0) &\propto \pi^*(a_0|D_0, D, \alpha_0, \beta_0)
&\equiv \int L_n(\beta|D)\frac{L_{n_0}(\beta|D_0)^{a_0} \pi_0(\beta)}{c^*(a_0)}a_0^{\alpha_0-1}(1-a_0)^{\beta_0-1}d\beta,
\end{align*}
and
\begin{align*}
\pi^*(a_0|D_0, D, \alpha_0, \beta_0) &\rightarrow a_0^{\alpha_0-1}(1-a_0)^{\beta_0-1}\int N(\hat{\beta}, \Sigma)N(\hat{\beta}_0, \Sigma_0)d\beta,\\
&\propto a_0^{\alpha_0-1}(1-a_0)^{\beta_0-1}\int\exp\left\{-\frac{1}{2}(\beta-\hat{\beta})'\Sigma^{-1}(\beta-\hat{\beta})\right\} \times \\
&\det(\Sigma_0)^{-\frac{1}{2}}\exp\left\{-\frac{1}{2}(\beta-\hat{\beta}_0)'\Sigma_0^{-1}(\beta-\hat{\beta}_0)\right\}d\beta,\\
&\text{(Assuming that } \hat{\beta}-\hat{\beta}_0=\delta)\\
&\propto a_0^{\alpha_0-1}(1-a_0)^{\beta_0-1}\det(\Sigma_0)^{-\frac{1}{2}}\det(\Sigma_n)^{\frac{1}{2}}\\
&\exp\left\{\frac{1}{2}\left[\hat{\beta}'\Sigma^{-1}\hat{\beta}-\delta'(\Sigma_0^{-1}-\Sigma_0^{-1}\Sigma_n\Sigma_0^{-1})\delta\right]\right\},
\end{align*}
where $\Sigma_n=(\Sigma^{-1}+\Sigma_0^{-1})^{-1}$.
Then 
\begin{align}
    &\pi(a_0|D_0, D, \alpha_0, \beta_0)\propto & \frac{a_0^{\alpha_0-1}(1-a_0)^{\beta_0-1}\det(\Sigma_0)^{-\frac{1}{2}}\det(\Sigma_n)^{\frac{1}{2}}\exp\big\{-\frac{1}{2}\delta'(\Sigma_0^{-1}-\Sigma_0^{-1}\Sigma_n\Sigma_0^{-1})\delta\big\}}{\int a_0^{\alpha_0-1}(1-a_0)^{\beta_0-1}\det(\Sigma_0)^{-\frac{1}{2}}\det(\Sigma_n)^{\frac{1}{2}}\exp\big\{-\frac{1}{2}\delta'(\Sigma_0^{-1}-\Sigma_0^{-1}\Sigma_n\Sigma_0^{-1})\delta\big\} da_0}.\label{glm_marg_a0}
\end{align}
We want to show that, if $\Sigma$ and $\Sigma_{0}$ are $p \times p$ positive definite matrices, 
\begin{align*}
    \lim\limits_{n\rightarrow\infty} \frac{\int_0^{\epsilon} a_{0}^{\alpha_0-1}(1-a_{0})^{\beta_0-1}\det(\Sigma_0)^{-\frac{1}{2}}\det(\Sigma_n)^{\frac{1}{2}}\exp\big\{-\frac{1}{2}\delta^T(\Sigma_0^{-1}-\Sigma_0^{-1}\Sigma_n\Sigma_0^{-1})\delta \big\}da_0}{\int_0^1 a_{0}^{\alpha_0-1}(1-a_{0})^{\beta_0-1}\det(\Sigma_0)^{-\frac{1}{2}}\det(\Sigma_n)^{\frac{1}{2}}\exp\big\{-\frac{1}{2}\delta^T(\Sigma_0^{-1}-\Sigma_0^{-1}\Sigma_n\Sigma_0^{-1})\delta\big\} da_0}=1,
\end{align*}
for $\delta \neq 0$ and $\epsilon > 0$.\\
We can write $\Sigma=n^{-1}P$ and $\Sigma_0=(nra_0)^{-1}P_0$ \citep{fahrmeir}, where $P$ and $P_0$ are positive definite and independent of $a_0$ and $n$.
Then $\Sigma_n=(\Sigma^{-1}+\Sigma_0^{-1})^{-1} = n^{-1}(P^{-1}+ra_0P_0^{-1})^{-1}$. \\
Now, $$I-\Sigma_n\Sigma_0^{-1}=I-(\Sigma^{-1}+\Sigma_0^{-1})^{-1}\Sigma_0^{-1}=(\Sigma^{-1}+\Sigma_0^{-1})^{-1}\Sigma^{-1},$$ and


\begin{align*}
\Sigma_0^{-1}(I-\Sigma_n\Sigma_0^{-1}) =&\Sigma_0^{-1}\Sigma_n\Sigma^{-1},\\
=&nra_0P_0^{-1}n^{-1}(P^{-1}+ra_0P_0^{-1})^{-1}nP^{-1},\\
=&nra_0P_0^{-1}(P^{-1}+ra_0P_0^{-1})^{-1}P^{-1},\\
=& nra_0(P_0 + a_0rP)^{-1},\\
=& nra_0P^{-1}[P_0P^{-1} + a_0rI]^{-1},\\
=& na_0P^{-1}[r^{-1}P_0P^{-1} + a_0I]^{-1}.\\
\end{align*}
In addition, 


\begin{align*}
\det(\Sigma_0)^{-\frac{1}{2}}\det(\Sigma_n)^{\frac{1}{2}},
=&\det((nra_0)^{-1}P_0)^{-\frac{1}{2}}\det(n^{-1}(P^{-1}+ra_0P_0^{-1})^{-1})^{\frac{1}{2}},\\
=&\det((ra_0)^{-1}P_0)^{-\frac{1}{2}}\det(P^{-1}+ra_0P_0^{-1})^{-\frac{1}{2}},\\
=&\det((ra_0)^{-1}P_0(P^{-1}+ra_0P_0^{-1}))^{-\frac{1}{2}},\\
=&\det(a_0^{-1}(r^{-1}P_0P^{-1}+a_0I))^{-\frac{1}{2}},\\
=&a_0^{\frac{p}{2}}\det(a_0I-(-r^{-1}P_0P^{-1}))^{-\frac{1}{2}}.\\
\end{align*}


Let $$h(a_0)=\frac{1}{2n}\delta^T(\Sigma_0^{-1}-\Sigma_0^{-1}\Sigma_n\Sigma_0^{-1})\delta=\frac{1}{2}\delta^Ta_0P^{-1}[r^{-1}P_0P^{-1} + a_0I]^{-1}\delta.$$ and $$f(a_0)=\det(\Sigma_0)^{-\frac{1}{2}}\det(\Sigma_n)^{\frac{1}{2}}=a_0^{\frac{p}{2}}\det(a_0I-(-r^{-1}P_0P^{-1}))^{-\frac{1}{2}}.$$
Then the denominator is
$$A=\int_0^1 a_0^{\alpha_0-1}(1-a_0)^{\beta_0-1}f(a_0)\exp\big\{-nh(a_0) \big\}da_0.$$
First, we show $h(a_0)$ is differentiable.\\
\begin{lem}\label{eigen}
Let $A$ and $B$ be positive definite matrices of the same dimension. Then, the eigenvalues of $AB$ are positive.
\end{lem}
\begin{proof}
By the spectral decomposition, $A=P\Lambda P^T$ where $\Lambda=diag(\lambda_1,\dots,\lambda_p)$ and $\lambda_1,\dots,\lambda_p$ are the eigenvalues of $A$. Then $A^{\frac{1}{2}}=P\Lambda^{\frac{1}{2}}P^T$ is symmetric $\Rightarrow v^TA^{\frac{1}{2}}BA^{\frac{1}{2}}v=(A^{\frac{1}{2}}v)^TB(A^{\frac{1}{2}}v) > 0$. So $A^{\frac{1}{2}}BA^{\frac{1}{2}}$ is positive definite. Since $A^{\frac{1}{2}}(A^{\frac{1}{2}}BA^{\frac{1}{2}})A^{-\frac{1}{2}} = AB$, $A^{\frac{1}{2}}BA^{\frac{1}{2}}$ and $AB$ are similar. Then they have the same eigenvalues and the eigenvalues of $AB$ are positive. 
\end{proof}
Let $B=a_0I - (-r^{-1}P_0P^{-1})$. Then $$h(a_0)=\frac{1}{2}a_0\delta^TP^{-1}B^{-1}\delta = \frac{\frac{1}{2}a_0\delta^TP^{-1}adj(B)\delta}{det(B)},$$ where $adj(B)$ is the cofactor matrix of $B$. The entries of $adj(B)$ are polynomials in $a_0$, so $\frac{1}{2}\delta^Ta_0P^{-1}adj(B)\delta$ is a polynomial in $a_0$ and thus differentiable. Then we show that $\det(B)^{-1}$ is differentiable on $(0,1)$. Since $\det(B)$ is a polynomial of $a_0$, it suffices to show that it is nonzero on $(0,1)$. Note that $\det(B)$ is the characteristic polynomial of $-r^{-1}P_0P^{-1}$. Since $P_0$ and $P^{-1}$ are positive definite, $-r^{-1}P_0P^{-1}$ has negative eigenvalues by Lemma \ref{eigen}. So $\det(B)$ is nonzero on $(0,1)$. Thus, we have shown  $h(a_0)$ is differentiable. \\\\
We then proceed to show that $h'(a) > 0$.


Let $E=P_0 + a_0rP$.
Then $h(a_0)=\frac{1}{2}a_0r\delta^TE^{-1}\delta$.
Therefore,
\begin{equation*}
    h'(a_0)=\frac{1}{2}r\delta^TE^{-1}\delta+a_0r\frac{1}{2}\delta^T(E^{-1})'\delta.
\end{equation*}
We know that $(E^{-1})'=E^{-1}E'E^{-1}=E^{-1}PE^{-1}$. Since $P$ is positive definite and $E$ is symmetric, $v^TE^{-1}PE^{-1}v=(E^{-1}v)^TPE^{-1}v > 0.$ Thus,  $E^{-1}PE^{-1}$ is positive definite. Then $a_0r\frac{1}{2}\delta^T(E^{-1})'\delta > 0$. Since $E^{-1}$ is positive definite, $\frac{1}{2}r\delta^TE^{-1}\delta > 0.$ So $h'(a_0) > 0$.\\
We also show $h'(a_0)$ is continuous.
It suffices to show that $\det(E)$ is nonzero on $[0,1]$. Since $E=rBP$ where $P$ is full rank, $\det(E)=c\det(B)$ where $c\neq0$. Since $\det(B)$ is nonzero, $\det(E)$ is also nonzero.\\
Next, we will show that $f(a_0)=a_0^{\frac{p}{2}}\det(a_0I-(-r^{-1}P_0P^{-1}))^{-\frac{1}{2}}=a_0^{\frac{p}{2}}\det(B)^{-\frac{1}{2}}$ is continuous on $[0,1]$. We have previously proven that $\det(B)$ is nonzero on $[0,1]$. Then $f(a_0)$ is continuous on $[0,1]$, and it will attain its minima and maxima on the closed interval. Let $t_1=\max_{[\epsilon, 1]}(f(a_0))$ and $t_2=\min_{[\frac{\epsilon}{2}, \epsilon]}(f(a_0))$. Since $a_{0}^{\alpha_0-1}(1-a_{0})^{\beta_0-1}$ is continuous on $[\frac{\epsilon}{2}, \epsilon]$, denote its minimum by $t_3$.\\
We write $A=A_1+A_2$ where 
\begin{align*}
 A_1 &= \int_0^{\epsilon}a_{0}^{\alpha_0-1}(1-a_{0})^{\beta_0-1}f(a_0)\exp(-nh(a_0))da_0   \quad \textrm{and}\\
 A_2 &= \int_{\epsilon}^1a_{0}^{\alpha_0-1}(1-a_{0})^{\beta_0-1}f(a_0)\exp(-nh(a_0))da_0.
\end{align*}
Now we want to show that $\lim_{n\rightarrow\infty}\frac{A_2}{A_1} = 0.$\\
First, we find the upper bound of $A_2$.
Since $h(a_0)$ is monotone increasing, $\exp(-nh(a_0)) \le \exp(-nh(\epsilon))$.
Since $f(a_0) \le t_1$, we have 
\begin{align*}
A_2 &\le t_1\exp(-nh(\epsilon))\int_{\epsilon}^1a_{0}^{\alpha_0-1}(1-a_{0})^{\beta_0-1}da_0,\\
&\le t_1\exp(-nh(\epsilon))\int_0^1a_{0}^{\alpha_0-1}(1-a_{0})^{\beta_0-1}da_0,\\
&= t_1\exp(-nh(\epsilon))\frac{\Gamma(\alpha_0)\Gamma(\beta_0)}{\Gamma(\alpha_0+\beta_0)},\\
&=C_1\exp(-nh(\epsilon)).
\end{align*}
Next, we find the lower bound of $A_1$.
We have previously shown that $h'(a_0)$ is continuous on $(0,1)$.
Then $h'(a_0)$ attains its maximum on $[\frac{\epsilon}{2}, \epsilon]$.
Let $t_4=\max_{[\frac{\epsilon}{2}, \epsilon]}(h'(a_0))$.
We can write 
\begin{align*}
A_1 &\ge \int_{\frac{\epsilon}{2}}^{\epsilon}a_{0}^{\alpha_0-1}(1-a_{0})^{\beta_0-1}f(a_0)\exp(-nh(a_0))da_0,\\
&\ge \int_{\frac{\epsilon}{2}}^{\epsilon}a_{0}^{\alpha_0-1}(1-a_{0})^{\beta_0-1}\frac{f(a_0)}{h'(a_0)}\exp(-nh(a_0))h'(a_0)da_0,\\
&= \int_{\frac{\epsilon}{2}}^{\epsilon}a_{0}^{\alpha_0-1}(1-a_{0})^{\beta_0-1}\frac{f(a_0)}{h'(a_0)}\exp(-nh(a_0))dh(a_0),\\
&\ge \frac{t_2t_3}{t_4} \int_{\frac{\epsilon}{2}}^{\epsilon}\exp(-nh(a_0))dh(a_0),\\
&= \frac{t_2t_3}{t_4}\frac{1}{n}[\exp(-nh(\epsilon/2)-\exp(-nh(\epsilon))],\\
&=C_2\frac{1}{n}[\exp(-nh(\epsilon/2)-\exp(-nh(\epsilon))].\\
\end{align*} 
Therefore, $$0 \le \frac{A_2}{A_1} \le \frac{C_1\exp(-nh(\epsilon))}{C_2\frac{1}{n}[\exp(-nh(\epsilon/2))-\exp(-nh(\epsilon))]}=\frac{C_1n}{C_2[\exp(-n[h(\epsilon/2)-h(\epsilon)])-1]},$$ 
and $\lim\limits_{n\rightarrow\infty} \frac{A_2}{A_1} \rightarrow 0$ by L'Hopital's rule. Since $\frac{A_2}{A_1} \ge \frac{A_2}{A}$, $\lim\limits_{n\rightarrow\infty} \frac{A_2}{A} \rightarrow 0$. Then $\lim\limits_{n\rightarrow\infty} \frac{A_1}{A} \rightarrow 1$. 
\end{proof}


\subsection{Proof of Corollary \ref{glm_coro}}\label{glm_coro_proof}

\begin{proof}
Based on the assumptions, we have $\Sigma = a_0\Sigma_0$. The result follows if we plug $\delta=0$ into \eqref{glm_marg_a0}. 
\end{proof}


\subsection{Proof for Theorem \ref{glm_comp}}\label{glm_comp_proof}

\begin{proof}
The Laplace approximation for multiple parameters has the form
\begin{align*}
\int\exp(-nf(\beta))d\beta \approx \exp(-nf(\hat{\beta}))\left(\frac{2\pi}{n}\right)^{p/2}|\hat{\Sigma}|^{1/2},
\end{align*}
where $\hat{\beta}$ is maximizes $-f(\beta)$, and $\hat{\Sigma}_{p\times p}=\left[\frac{\partial^2f(\hat{\beta})}{\partial{\beta_j}\partial{\beta_k}}\right]^{-1}$. \\
When $X'Y=X_0'Y_0$ and $X \neq X_0$, 

\begin{align*}
\pi(a_0|D,D_0) &\propto \int L_n(D|\beta)\frac{L_{n_0}(D_0|\beta)^{a_0}\pi_0(\beta)}{c(a_0)}\pi(a_0)d\beta,\\
&=\int L_n(D|\beta)\frac{\exp\left(a_0\left[\sum_{i=1}^{n}y_{i}x_{i}'\beta-\sum_{i=1}^nb(x_{0i}'\beta)\right]\right)}{\int \exp\left(a_0\left[\sum_{i=1}^{n}y_{i}x_{i}'\beta-\sum_{i=1}^nb(x_{0i}'\beta)\right]\right)d\beta}\pi(a_0)d\beta,\\
&=\pi(a_0)\frac{\int L_n(D|\beta)L_{n_0}^*(D_0|\beta, a_0)d\beta}{\int L_{n_0}^*(D_0|\beta, a_0)d\beta},\\
&=\pi(a_0)\frac{c_1(a_0)}{c_2(a_0)}.\\
\end{align*}
Define
\begin{align*}
g_n(\beta)&=-\frac{1}{n}[L_n(D|\beta)+a_0L_{n_0}^*(D_0|\beta, a_0)]\\
&=-\frac{1}{n}\{\log(Q(Y)) + \sum_{i=1}^{n}y_{i}x_{i}'\beta-\sum_{i=1}^nb(x_{i}'\beta) + a_0[\sum_{i=1}^{n}y_{i}x_{i}'\beta-\sum_{i=1}^nb(x_{0i}'\beta)]\}\\
&=-\frac{1}{n}\{\log(Q(Y)) + (a_0+1)\sum_{i=1}^{n}y_{i}x_{i}'\beta-\sum_{i=1}^nb(x_{i}'\beta) - a_0\sum_{i=1}^nb(x_{0i}'\beta)\}.
\end{align*}
Then we have
\begin{align*}
c_1(a_0) &\approx \exp(-ng_n(\hat{\beta}))\left(\frac{2\pi}{n}\right)^{p/2}|\hat{\Sigma}_g|^{1/2},
\end{align*}
where $\hat{\beta}$ maximizes $-g_n(\beta)$.
Similarly, define
\begin{align*}
k_n(\beta)&=-\frac{1}{n}a_0l^*(y|x_0, \beta),\\
&=-\frac{1}{n}\left\{a_0\sum_{i=1}^{n}y_{i}x_{i}'\beta - a_0\sum_{i=1}^nb(x_{0i}'\beta)\right\}.
\end{align*}
Then we have
\begin{align*}
c_2(a_0) &\approx \exp(-nk_n(\tilde{\beta}))\left(\frac{2\pi}{n}\right)^{p/2}|\tilde{\Sigma}_g|^{1/2},
\end{align*}
where $\tilde{\beta}$ maximizes $-k_n(\beta)$. \\
We compute the gradients of $g_n(\beta)$ and $k_n(\beta)$ and get
\begin{align*}
\nabla g_n(\beta) &=-\frac{1}{n}\{(a_0+1)\sum_{i=1}^{n}y_{i}x_{i}-\sum_{i=1}^n\dot{b}(x_{i}'\beta)x_i - a_0\sum_{i=1}^n\dot{b}(x_{0i}'\beta)x_{0i}\},\\
\nabla k_n(\beta) &=-\frac{1}{n}\{a_0\sum_{i=1}^{n}y_{i}x_{i}- a_0\sum_{i=1}^n\dot{b}(x_{0i}'\beta)x_{0i}\},\\
\nabla g_n(\beta)=0 &\Rightarrow \sum_{i=1}^n\dot{b}(x_{i}'\hat{\beta})x_i + a_0\sum_{i=1}^n\dot{b}(x_{0i}'\hat{\beta})x_{0i}=(a_0+1)\sum_{i=1}^{n}y_{i}x_{i},\\
\nabla k_n(\beta)=0 &\Rightarrow \sum_{i=1}^n\dot{b}(x_{0i}'\tilde{\beta})x_{0i}=\sum_{i=1}^{n}y_{i}x_{i}.
\end{align*}
We can see that asymptotically, $\hat{\beta} \neq \tilde{\beta}$. Then we have
\begin{align}
\frac{c_1(a_0)}{c_2(a_0)} &= \frac{|\hat{\Sigma}_g|^{1/2}}{|\tilde{\Sigma}_k|^{1/2}}\exp\{-n[g_n(\hat{\beta})-k_n(\tilde{\beta})]\},\label{c1}
\end{align}
where
\begin{align*}
\hat{\Sigma}_g&=\left[\frac{1}{n}\sum_{i=1}^{n}\ddot{b}(x_{i}'\hat{\beta})x_ix_i' + \frac{a_0}{n}\sum_{i=1}^{n_0}\ddot{b}(x_{0i}'\hat{\beta})x_{0i}x_{0i}'\right]^{-1},\\
\tilde{\Sigma}_k &=\left[\frac{a_0}{n}\sum_{i=1}^{n_0}\ddot{b}(x_{0i}'\tilde{\beta})x_{0i}x_{0i}'\right]^{-1},\\
\frac{|\hat{\Sigma}_g|^{1/2}}{|\tilde{\Sigma}_k|^{1/2}}&=\frac{|a_0\sum_{i=1}^{n_0}\ddot{b}(x_{0i}'\tilde{\beta})x_{0i}x_{0i}'|^{1/2}}{|\sum_{i=1}^{n}\ddot{b}(x_{i}'\hat{\beta})x_ix_i' + a_0\sum_{i=1}^{n_0}\ddot{b}(x_{0i}'\hat{\beta})x_{0i}x_{0i}'|^{1/2}}.
\end{align*}
The marginal posterior of $a_0$ is then proportional to \ref{c1} multiplied by $\pi(a_0)$.
\end{proof}

\newpage

\section{Additional Figures for Section 2}\label{app_sec2}

\begin{figure}[H]
\begin{center}
\includegraphics[width=12cm]{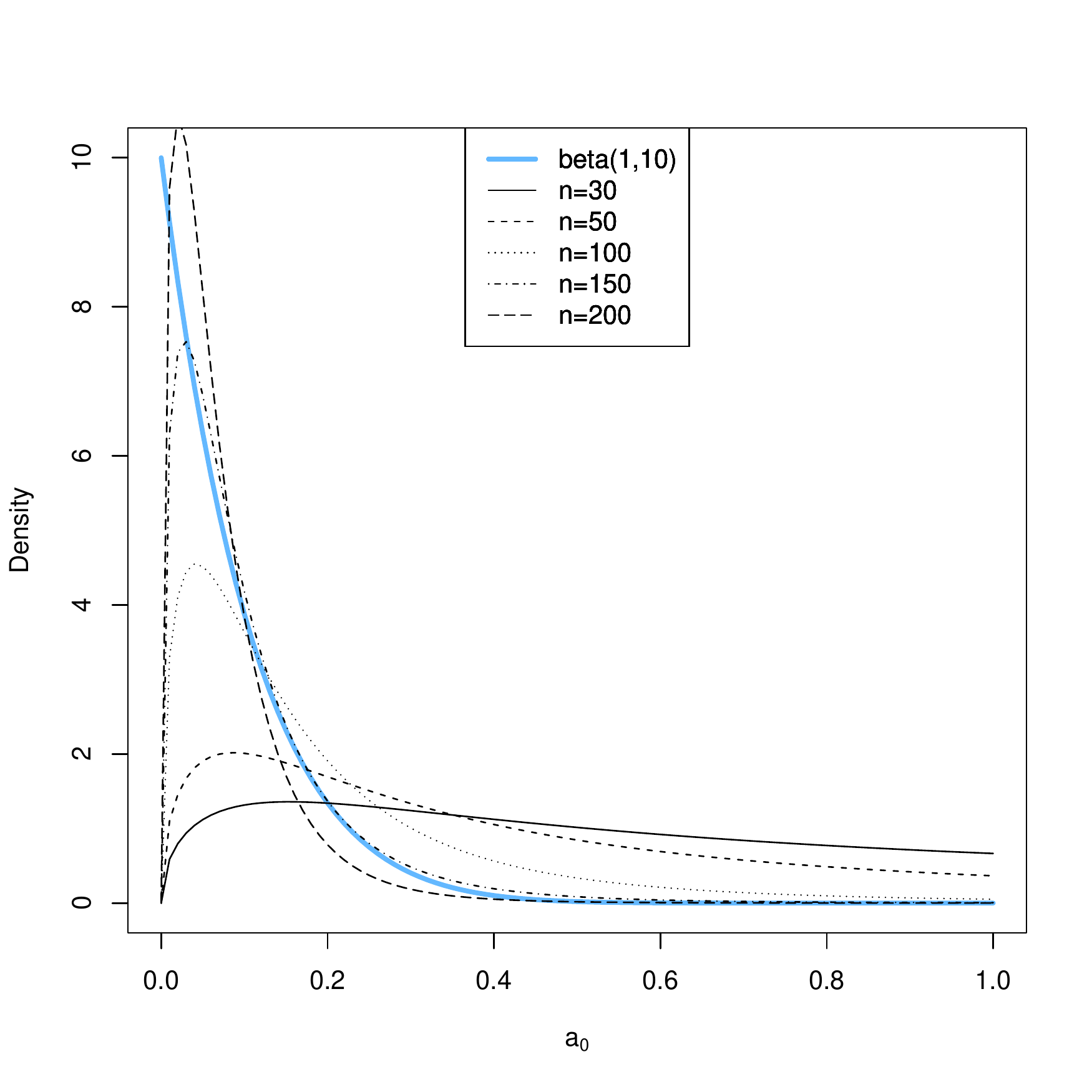}
\end{center}
\caption{Marginal posterior of $a_0$ for i.i.d. normal data where $n=n_0$ increases from $30$ to $200$, the historical data mean is $1.5$, the current data mean is $2$ and the standard deviations are 1. We observe that when there is some difference between the sufficient statistics of the historical and current data, the marginal posterior of $a_0$ converge to a point mass at zero quickly.}
\label{app_fig_marg}
\end{figure}

\begin{figure}[H]
\begin{center}
\includegraphics[width=14cm]{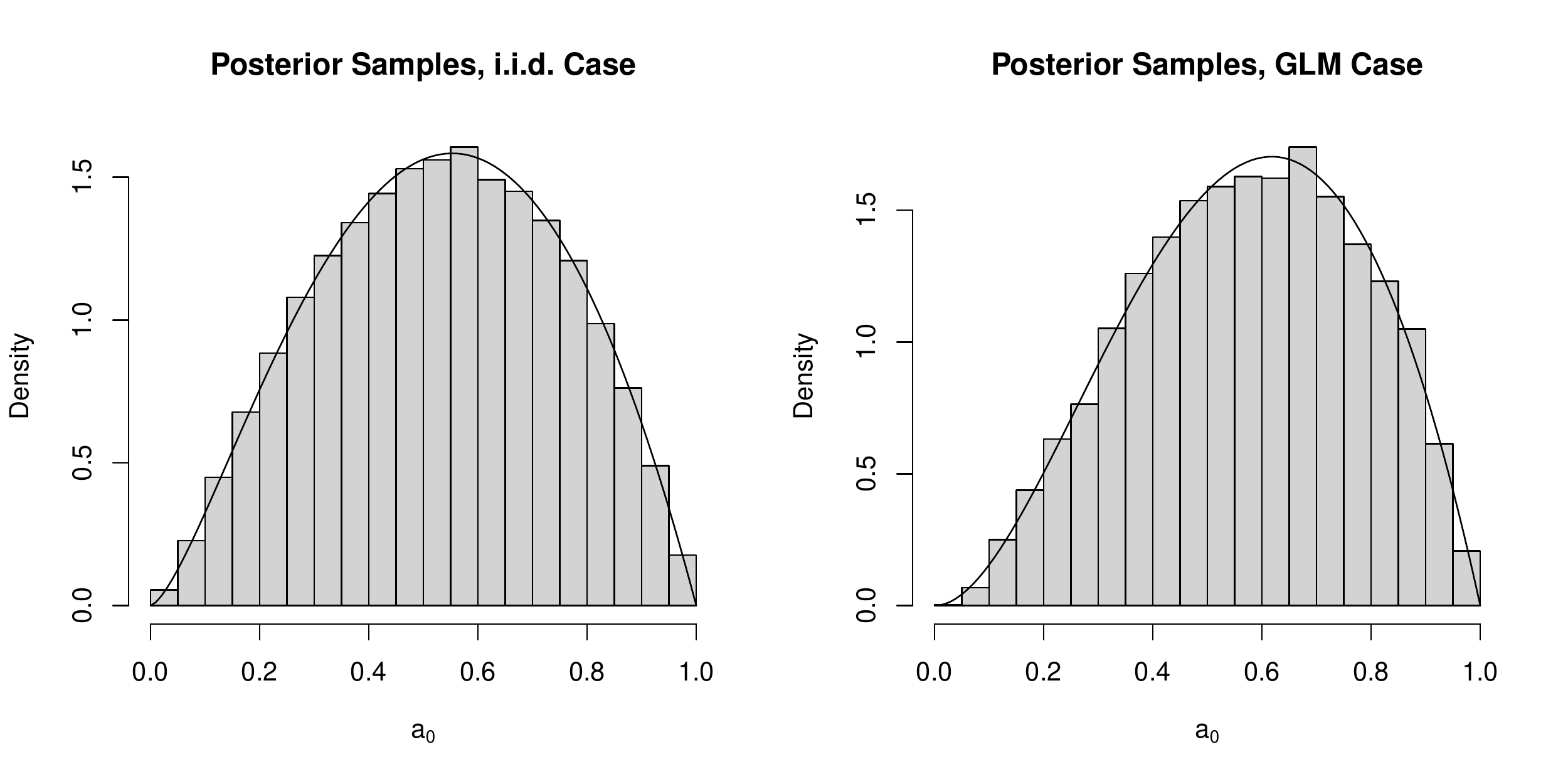}
\end{center}
\caption{The plot on the left shows the histogram of the posterior of $a_0$ for i.i.d. Bernoulli data with current and historical mean equal to 0.7, $n=100$, $n_0=200$ and the prior on $a_0$ is beta(2, 2). The plot on the right shows the histogram of the posterior of $a_0$ for Bernoulli data with one covariate where the historical and current data are identical. The prior on $a_0$ is beta(2, 2). The histograms of the posterior samples are produced using R package BayesPPD. The curve represents the theoretical density. We observe that for both i.i.d. and GLM cases, the histograms of posterior samples agree with the theoretical density functions. }
\label{app_fig_coro}
\end{figure}

\section{Numerical Stability of the Optimization Process}\label{app_num}

We conduct a simple experiment to reproduce the optimal priors derived in Figure 1 when we fix one of $\alpha_0$ or $\beta_0$ and optimize for the other parameter. We can see in Table \ref{app_tab_num} below that the optimization is stable and reliable. 
\begin{table}[H]
\begin{center}
\caption{Optimization with one of the hyperparameters fixed}\label{app_tab_num}
\begin{tabular}{llll}
\hline
& \thead{optimal priors\\ in Fig. 1} & \thead{optimal priors\\with fixed $\alpha_0$} & \thead{optimal priors \\with fixed $\beta_0$}\\
\hline
$d_{MTD}=0.5$ & beta(2.2, 2.3)& beta(2.2, 2.3)& beta(2.2, 2.3)\\
$d_{MTD}=1$ & beta(1, 0.4)& beta(1, 0.5)&beta(0.9, 0.4)\\
$d_{MTD}=1.5$ & beta(2.6, 0.5)&beta(2.6, 0.5)&beta(2.6, 0.5)\\
\hline
\end{tabular}
\end{center}
\end{table}

\section{Bias and Variance Decomposition for the MSE Criterion}\label{app_bias}

\begin{table}[H]
\begin{center}
\caption{Bias and variance decomposition for different prior choices}\label{app_tab_mse}
\begin{tabular}{lccc}
& Optimal Prior & Beta$(1,1)$ & Beta$(2,2)$ \\
\hline
Bias$^2$\\[5pt]
$d_{\textrm{MTD}}=0.5$ & 0.011 & 0.015 & 0.018  \\
$d_{\textrm{MTD}}=1$ & 0.005 & 0.012 & 0.025 \\
$d_{\textrm{MTD}}=1.5$ & 0.003 & 0.006 & 0.015  \\[5pt]
Variance&&&\\[5pt]
$d_{\textrm{MTD}}=0.5$ & 0.043 & 0.042 & 0.039  \\
$d_{\textrm{MTD}}=1$ & 0.058 & 0.057 & 0.054 \\
$d_{\textrm{MTD}}=1.5$ & 0.049 & 0.053 & 0.052 \\
\end{tabular}
\end{center}
\end{table}

\section{Comparisons with Other Priors}\label{app_rmap}
In Figure \ref{app_fig_rmap}, we generate i.i.d. normal data and compute the MSE based on the posterior mean of the point estimator using three different prior choices, the NPP with the optimal beta prior on $a_0$ (optimal in the sense of minimizing MSE as defined in the main paper), the NPP with a mixture of two beta priors on $a_0$, and the robust mixture prior, which is a special case of the robust meta-analytic-predictive prior introduced in \cite{Schmidli_2014}. The robust mixture prior we use places equal weights on a non-informative normal component and an informative normal component using the historical data. For the NPP with a mixture of beta priors, we use a mixture of beta(1, $c$) and beta($c$, 1) with equal weights, where $c$ ranges from $100$ to $1000$. As $c$ approaches infinity, the mixture of beta priors on $a_0$ converges to a mixture of a point mass at zero and a point mass at one, which is equivalent to the robust mixture prior. We vary the difference between the observed current data mean and the historical data mean, i.e., $d_{obs}=\bar{y}_{obs}-\bar{y}_0$. We can see that when the data are compatible ($d_{obs} = 0.5$), the posterior mean based on the NPP with the optimal beta prior produces lower MSE than the estimator based on the robust mixture prior. When the conflict between the data increases, i.e., $d_{obs}=1$ and $d_{obs}=1.5$, the NPP with a mixture of beta priors outperforms the robust mixture prior.
\begin{figure}[H]
\begin{center}
\includegraphics[width=15cm]{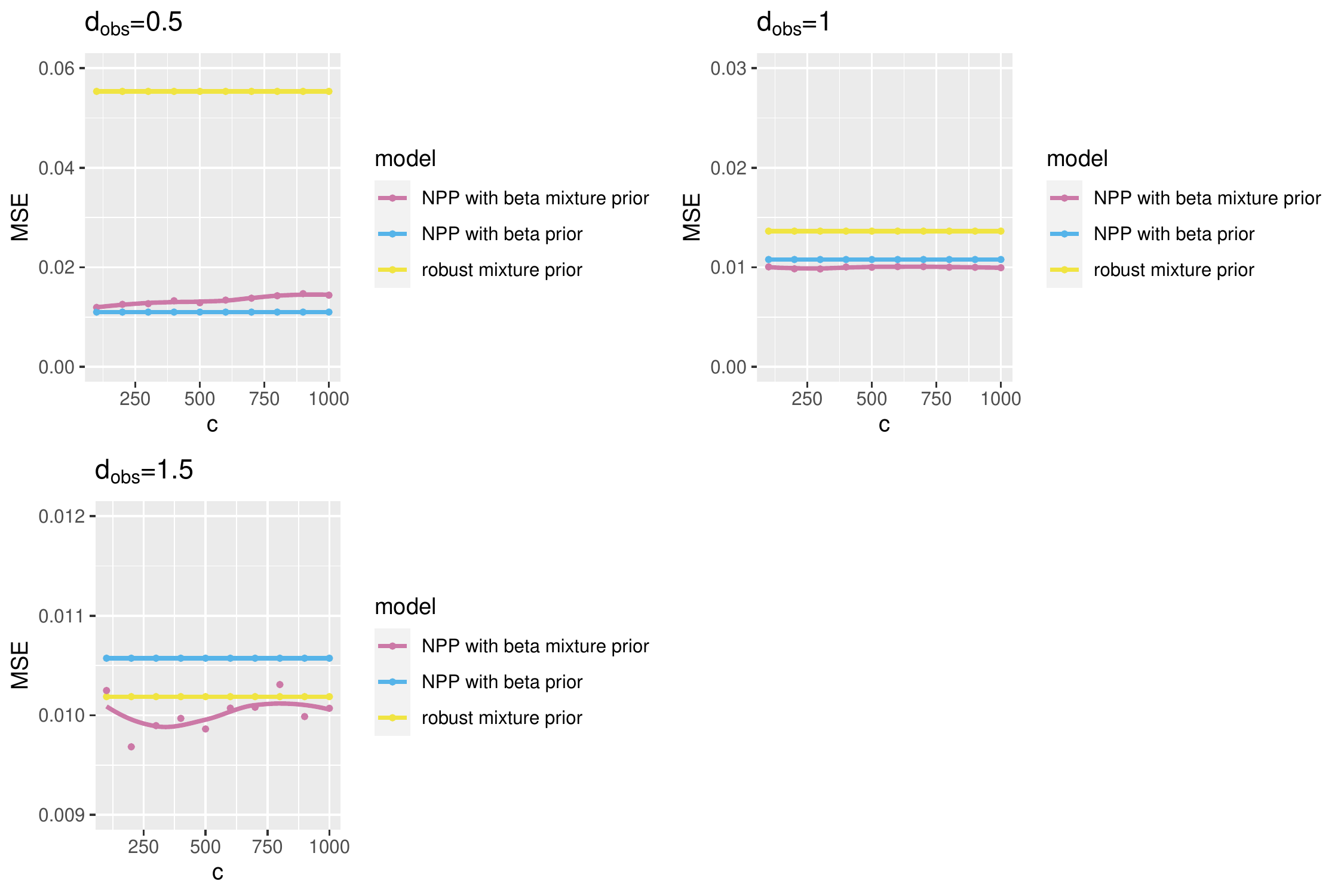}
\end{center}
\caption{MSE using three different prior choices, the NPP with the optimal beta prior on $a_0$, the NPP with the mixture of beta priors on $a_0$ and the robust mixture prior}
\label{app_fig_rmap}
\end{figure}

\section{Additional Simulations for MSE Criterion}\label{app_mse}

Figures \ref{app_fig_mse1} and \ref{app_fig_mse2} show the  MSE as a function of the prior mean of $a_0$ for increasing ratios of $n / n_0$ when the total sample size is fixed. We observe that as $n / n_0$ increases, the model will increasingly benefit, i.e. the MSE is reduced, from borrowing more, but this trend is less prominent when the total sample size is larger. 

The total sample size of the PLUTO trials in section \ref{sec:ped} is about twice the total sample size of the melanoma trials in section \ref{sec:mel}. The total sample size of the melanoma trials is not large enough for the model to criticize the maximal tolerable difference that we chose. Therefore, the optimal prior derived using the MSE criterion encourages borrowing for the melanoma trial.

\begin{figure}[H]
\begin{center}
\includegraphics[width=14cm]{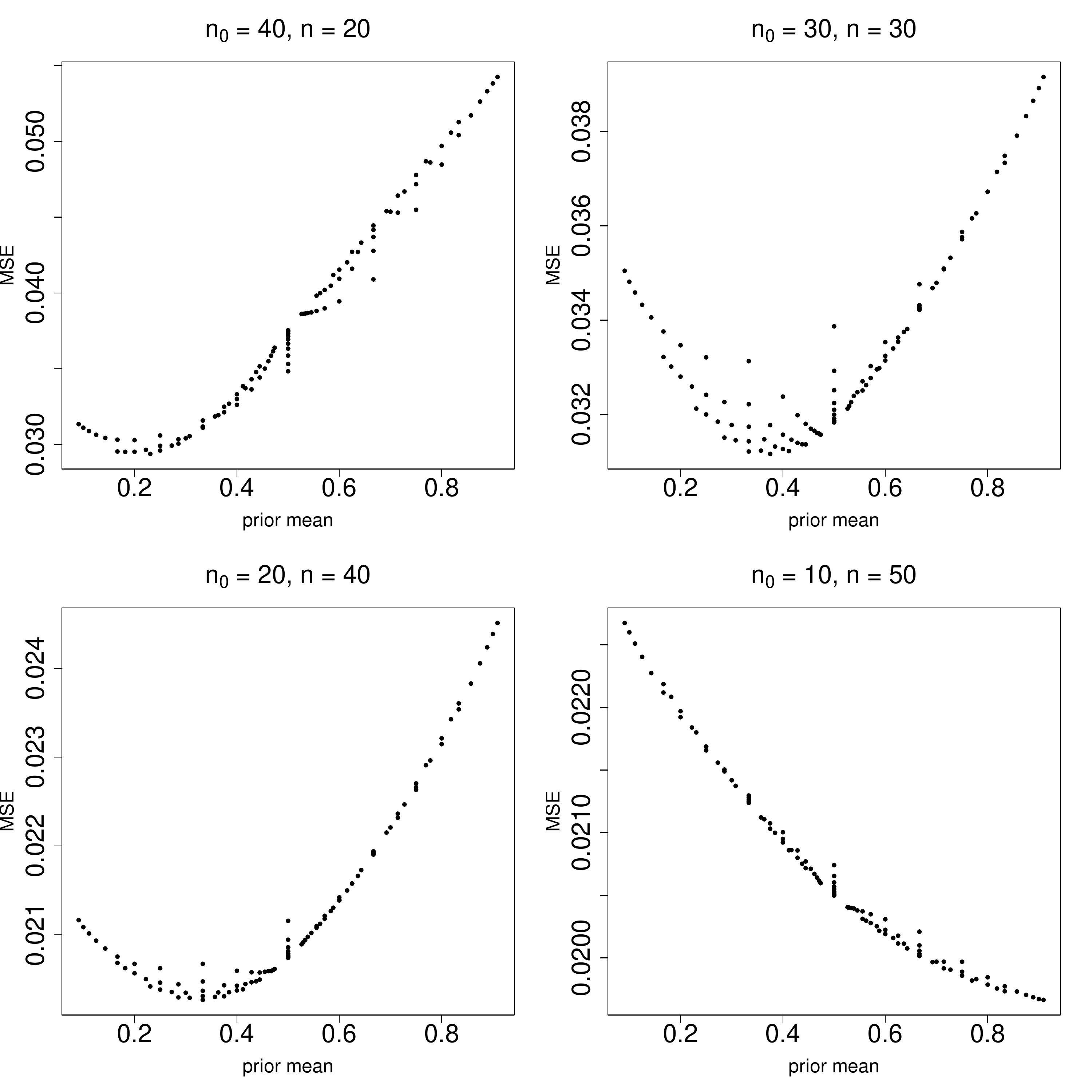}
\end{center}
\caption{MSE as a function of prior mean of $a_0$ for increasing ratios of $n / n_0$ when the total sample size is fixed for the normal \textit{i.i.d.} case.}
\label{app_fig_mse1}
\end{figure}

\begin{figure}[H]
\begin{center}
\includegraphics[width=14cm]{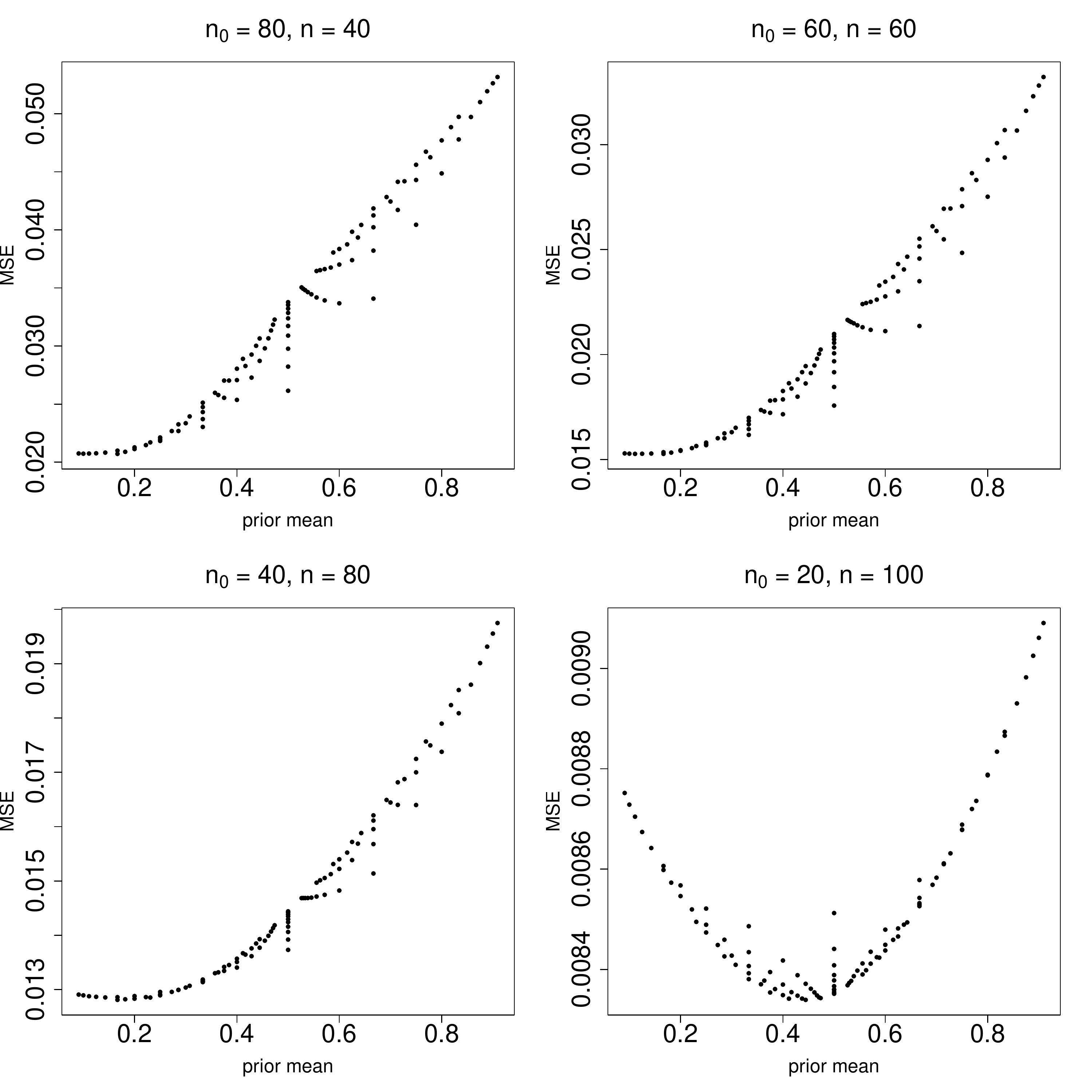}
\end{center}
\caption{MSE as a function of prior mean of $a_0$ for increasing ratios of $n / n_0$ when the total sample size is double the total sample size in Figure \ref{app_fig_mse1} for the normal \textit{i.i.d.} case.}
\label{app_fig_mse2}
\end{figure}

\section{Design Application for the Pediatric Lupus Trial}\label{app_design}
Now we demonstrate using the proposed optimal priors in a clinical trial design application. Suppose we want to design a pediatric trial using data from the adult trials BLISS-52 and BLISS-76. We choose a few sample sizes ranging from $50$ to $100$ (the actual trial had a sample size of $92$) and derive the optimal prior for each sample size using both the KL and MSE criteria. We compute power using the R Package \emph{BayesPPD} which performs Bayesian sample size determination with a simulation-based procedure \citep{shen_RJ}. We use the posterior samples given only the historical
data as the discrete approximation to the sampling prior \citep{Psioda_Ibrahim_2019}. For the fitting prior, we use a normalized power prior with optimal priors derived for $a_0$. Figure \ref{app_fig_ped_power} shows the power curves for three choices of priors on $a_0$, the optimal prior using the KL criterion, the optimal prior using the MSE criterion, and the uniform prior. Note the optimal prior is derived for each sample size. In this case, the optimal priors do not vary much for different sample sizes due to the small sizes of the current trial relative to the adult trials. We can see that power is the highest when we optimize to minimize KL. Since the optimal prior on $a_0$ based on the KL criterion is beta($5.5$, $5.5$) (when $n=100$), the most amount of historical information is borrowed. Power is the lowest when we optimize to minimize MSE, since the least amount of historical information is borrowed. The two criteria address the problem of how much to borrow from different angles. The KL criterion focuses on how much information one is willing to borrow under two markedly different assumptions about the difference between the prior information and the data generation process for the future study. The KL criterion does not explicitly focus on estimation performance metric. On the other hand, the MSE criterion attempts to ensure that the point estimate of the parameter of interest behaves well in terms of the trade-off between bias and variance. Also note that the two target distributions for the KL criteria used in this  application are beta(1, 10) and beta(10, 1). These target distributions should be carefully chosen so that they reflect the desired posterior distributions of $a_0$ relative to the sample sizes of the historical and current data. For example, by considering $c=10$ one is targeting borrowing approximately $10\%$ of the prior information when the prior-data conflict is substantial (i.e., in line with $d_{MTD}$). If $10\%$ of the historical data sample size is large relative to the new study sample size being considered, this choice for $c$ may not be desirable (i.e., a larger $c$ would be warranted).

\begin{figure}[H]
\centering
\includegraphics[width=13cm]{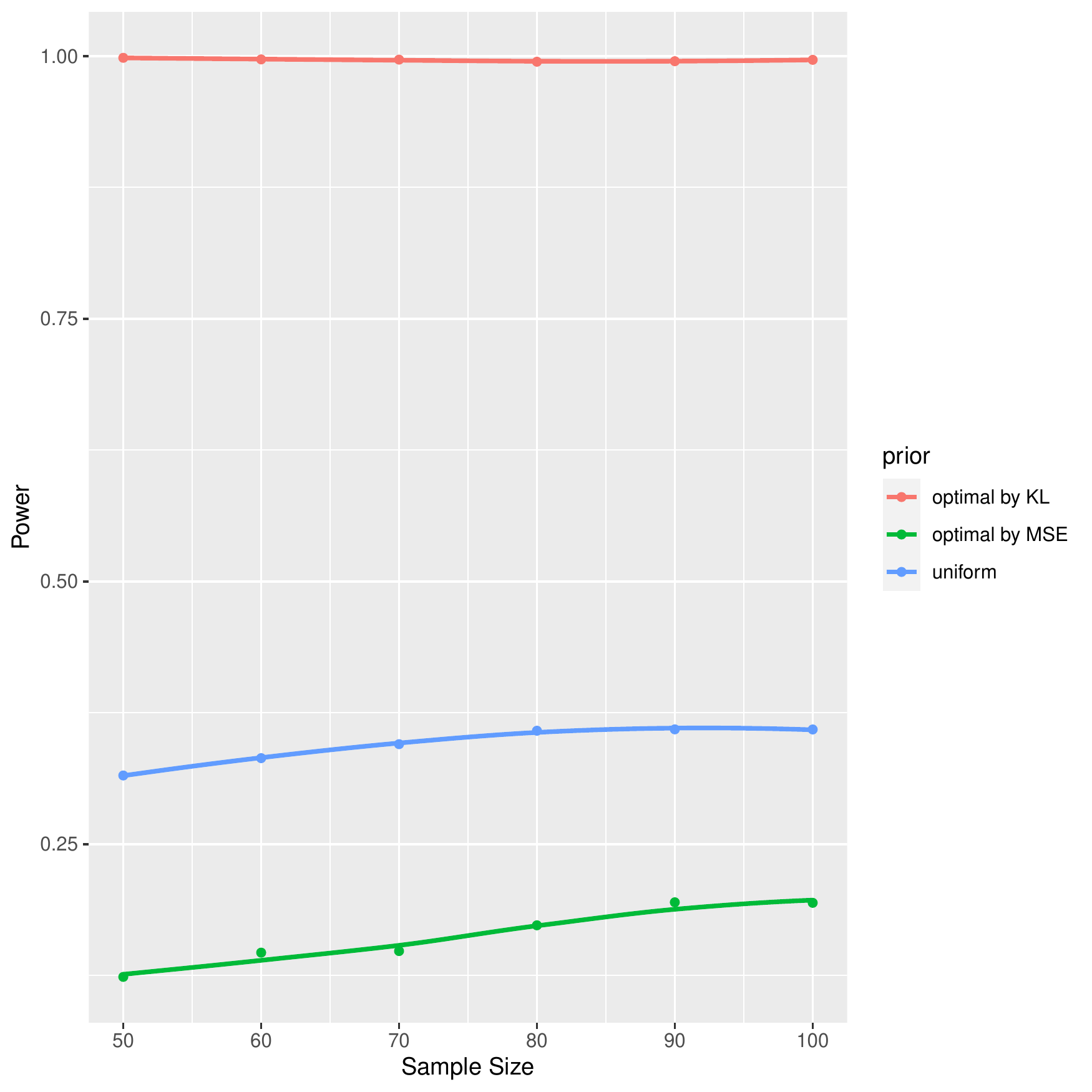}
\caption{Power curves using three choices of priors on $a_0$, optimal prior using the KL criterion, optimal prior using the MSE criterion, and the uniform prior. A different optimal prior is derived for each sample size.}
\label{app_fig_ped_power}
\end{figure}

\end{appendix}

\section*{Acknowledgements}
We thank Professor Yuri Saporito and Rodrigo Alves for valuable discussions.

\bibliographystyle{biom}
\bibliography{refs}

\end{document}